\documentclass[]{article}
\usepackage{graphicx}
\usepackage{cite}
\usepackage{color}
\usepackage{amsfonts}
\usepackage{amsthm}
\usepackage{amsmath}
\usepackage{amssymb}
\usepackage{subfig}
\usepackage{fixltx2e}

\usepackage{pseudocode}
\newtheorem{theorem}{Theorem}
\newtheorem{lemma}{Lemma}
\newtheorem{remark}{Remark}
\newtheorem{definition}{Definition}
\newtheorem{corollary}{Corollary}

\newtheorem{proposition}{Proposition}
\newtheorem{example}{Example}

\begin{document}

\title{Multiterminal Source Coding \\ under Logarithmic Loss
%\thanks{A small part of this work was presented at the IEEE International Symposium on Information Theory, Saint Petersburg, Russia, August, 2011.}
}


\author{Thomas A.\ Courtade\thanks{Thomas Courtade is with the Department of Electrical Engineering, 
University of California, Los Angeles. Email: \texttt{tacourta@ee.ucla.edu}.}
~and~Tsachy~Weissman\thanks{Tsachy Weissman is with the Department of Electrical Engineering, 
Stanford University.
Email: \texttt{tsachy@stanford.edu}.}
}

\date{\today}

\maketitle

\begin{abstract}
We consider the classical two-encoder multiterminal source coding problem where distortion is measured under logarithmic loss.  We provide a single-letter description of the achievable rate distortion region for arbitrarily correlated sources with finite alphabets. In doing so, we also give the rate distortion region for the $m$-encoder CEO problem (also under logarithmic loss).  Several applications and examples are given.
\end{abstract}

%\begin{IEEEkeywords}
%Keywords.
%\end{IEEEkeywords}

\section{Introduction}

%\subsection{Background}
A complete characterization of the achievable rate distortion region for the two-encoder source coding problem depicted in Figure \ref{fig:MTSC} has remained an open problem for over three decades.  Following tradition, we will refer to this two-encoder source coding network as the \emph{multiterminal source coding problem} throughout this paper.  Several special cases have been solved for general source alphabets and distortion measures:
\begin{itemize}
\item The lossless case where $D_1=0,D_2=0$.  Slepian and Wolf solved this case in their seminal work\cite{bib:SlepianWolf1973}.
\item The case where one source is recovered losslessly: i.e., $D_1=0,D_2=D_{\mbox{max}}$.  This case corresponds to the source coding with side information problem of Ahlswede-K\"{o}rner-Wyner \cite{bib:AhlswedeKorner1975},\cite{bib:Wyner1975}.
\item The Wyner-Ziv case \cite{bib:WynerZiv1976} where $Y_2$ is available to the decoder as side information and $Y_1$ should be recovered with distortion at most $D_1$.
\item The Berger-Yeung case (which subsumes the previous three cases) \cite{bib:BergerYeung1989} where $D_1$ is arbitrary and $D_2=0$.
\end{itemize}

\begin{figure}
\begin{center}
\def\svgwidth{4in}
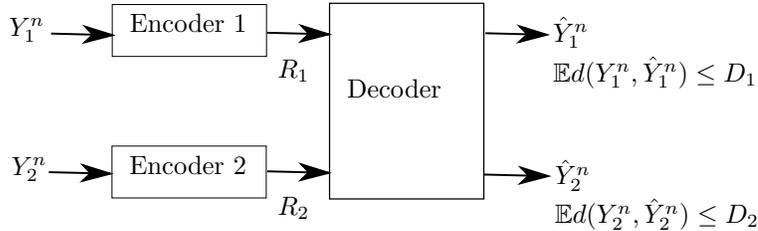
\caption{The multiterminal source coding network.}\label{fig:MTSC}
\end{center}
\end{figure}

Despite the apparent progress, other seemingly fundamental cases, such as when $D_1$ is arbitrary and $D_2=D_{\mbox{max}}$, remain unsolved except perhaps in very special cases.

Recently, the achievable rate distortion region for the quadratic Gaussian
multiterminal source coding problem was given by Wagner, Tavildar, and \linebreak Viswanath \cite{bib:Wagner2008}.    Until now, this was the only case for which the entire achievable rate distortion region was known.  While this is a very important result, it is again a special case from a theoretical point of view: a specific choice of source distribution, and a specific choice of distortion measure.

In the present paper, we determine the achievable rate distortion region of the multiterminal source coding problem for arbitrarily correlated sources with finite alphabets.  However, as in \cite{bib:Wagner2008}, we restrict our attention to a specific distortion measure.

At a high level, the roadmap for our argument is similar to that of \cite{bib:Wagner2008}.  In particular, both arguments couple the multiterminal source coding problem to a parametrized family of CEO problems.  Then, the parameter in the CEO problem is ``tuned" to yield the converse result.  Despite this apparent similarity, the proofs in \cite{bib:Wagner2008} rely heavily on the previously known Gaussian CEO results \cite{bib:PrabhakaranTseRamachandran2004}, the Gaussian one-helper results \cite{bib:Oohama1997}, and the calculus performed on the closed-form entropy expressions which arise from the Gaussian source assumption.  In our case we do not have this luxury, and our CEO tuning argument essentially relies on an existence lemma to yield the converse result.  The success of our approach is largely due to the fact that the distortion measure we consider admits a lower bound in the form of a conditional entropy, much like the quadratic distortion measure for Gaussian sources.

\subsection{Our Contributions}
In this paper, we give a single-letter characterization of the achievable rate distortion region  for the multiterminal source coding problem under logarithmic loss.  In the process of accomplishing this, we derive the achievable rate distortion region for the $m$-encoder CEO problem, also under logarithmic loss.  In both settings, we obtain a stronger converse than is standard for rate distortion problems in the sense that augmenting the reproduction alphabet does not enlarge the rate distortion region.  Notably, we make no assumptions on the source distributions, other than that the sources have finite alphabets. In both cases, the Berger-Tung inner bound on the rate distortion region is tight.  To our knowledge, this constitutes the first time that the entire achievable rate distortion region has been described for general finite-alphabet sources under nontrivial distortion constraints.  

\subsection{Organization} 
This paper is organized as follows.  In Section \ref{sec:MTSCprobDef} we formally define the logarithmic loss function and the multiterminal source coding problem we consider. In Section \ref{sec:ceoProblem} we define the CEO problem and give the rate distortion region under logarithmic loss.  In Section \ref{sec:MTSC} we return to the multiterminal source coding problem and derive the rate distortion region for the two-encoder setting.  Also in Sections \ref{sec:ceoProblem} and \ref{sec:MTSC}, applications to estimation, horse racing, and list decoding are given. In Section \ref{sec:generalMTSC}, we discuss connections between our results and the multiterminal source coding problem with arbitrary distortion measures.  Section \ref{sec:Conc} delivers our concluding remarks and discusses directions for future work.

\section{Problem Definition}\label{sec:MTSCprobDef}
%We study the multiterminal source coding problem subject to distortion constraints computed using a specific distortion measure and associated reproduction alphabet.  In this section, we formally define the system model and the logarithmic loss distortion measure we consider.

Throughout this paper, we adopt notational conventions that are standard in the literature.  Specifically, random variables are denoted by capital letters (e.g., $X$) and their corresponding alphabets are denoted by corresponding calligraphic letters (e.g., $\mathcal{X}$).  We abbreviate a sequence $(X_1,X_2,\dots,X_n)$ of $n$ random variables by $X^n$, and we denote the interval $(X_k,X_{k+1},\dots,X_j)$ by $X_k^j$.  If the lower index is equal to $1$, it will be omitted when there is no ambiguity  (e.g., $X^j \triangleq X_1^j$).  Frequently, random variables will appear with two subscripts (e.g., $Y_{i,j}$).  In this case, we are referring to the $j^{th}$ instance of random variable $Y_i$.  We overload our notation here slightly in that $Y_{i,1}^j$ is often abbreviated as $Y_i^j$.  However, our meaning will always be clear from context.

Let $\left\{(Y_{1,j},Y_{2,j})\right\}_{j=1}^n = (Y_{1}^n,Y_{2}^n)$ be a sequence of $n$ independent, identically distributed random variables with finite alphabets $\mathcal{Y}_1$ and $\mathcal{Y}_2$ respectively and joint pmf $p(y_1,y_2)$.  That is, $(Y_{1}^n,Y_{2}^n)\sim \prod_{i=1}^n
p(y_{1,j},y_{2,j})$.

In this paper, we take the reproduction alphabet $\hat{\mathcal{Y}}_i$ to be equal to the set of probability distributions over the source alphabet $\mathcal{Y}_i$ for $i=1,2$.  Thus, for a vector $\hat{Y}_i^n \in \hat{\mathcal{Y}}_i^n$, we will use the notation $\hat{Y}_{i,j}(y_i)$ to mean the $j^{th}$ coordinate ($1\leq j \leq n$) of $\hat{Y}_i^n$ (which is a probability distribution on $\mathcal{Y}_i$) evaluated for the outcome $y_i\in\mathcal{Y}_i$.  In other words, the decoder generates `soft' estimates of the source sequences.

We will consider the logarithmic loss distortion measure defined as follows:
\begin{align*}
d({y}_i,\hat{y}_i)=\log\left(\frac{1}{\hat{y}_i(y_i)} \right)  = D(1_{y_i}(y) \| \hat{y}_i(y) )\mbox{~for $i=1,2$}.
\end{align*}
In particular, $d({y}_i,\hat{y}_i)$ is the relative entropy (i.e., Kullback-Leibler divergence) between the empirical distribution of the event $\{Y_i=y_i\}$ and the estimate $\hat{y}_i$.
Using this definition for symbol-wise distortion, it is standard to define the distortion between sequences as 
\begin{align*}
d({y}_i^n,\hat{y}_i^n)=\frac{1}{n}\sum_{j=1}^n d({y}_{i,j},\hat{y}_{i,j}) \mbox{~for $i=1,2$}.
\end{align*}

We point out that the logarithmic loss function is a widely used penalty function in the theory of learning and prediction (cf. \cite[Chapter 9]{bib:Cesa-Bianchi:2006}).  Further, it is a particularly natural loss criterion in settings where the reconstructions are allowed to be `soft', rather than deterministic values.  Surprisingly, since distributed learning and estimation problems are some of the most oft-cited applications of lossy multiterminal source coding, it does not appear to have been studied in this context until the recent work \cite{bib:CourtadeISIT2011}. However, we note that this connection has been established previously for the single-encoder case in the study of the information bottleneck problem \cite{bib:Harremoes2007}.  Beyond learning and prediction, a similar distortion measure has appeared before in the image processing literature \cite{bib:Andre2006}. As we demonstrate through several examples, the logarithmic loss distortion measure has a variety of useful applications in the context of multiterminal source coding.

A rate distortion code (of blocklength $n$) consists of encoding functions:
\begin{align*}
g_i^{(n)}: \mathcal{Y}_i^n \rightarrow \left\{ 1,\dots,M_i^{(n)}\right\} \mbox{~for $i=1,2$}
\end{align*}
and decoding functions
\begin{align*}
\psi_i^{(n)}: \left\{ 1,\dots,M_1^{(n)} \right\} \times \left\{ 1,\dots,M_2^{(n)} \right\} \rightarrow \hat{\mathcal{Y}}_i^n \mbox{~for $i=1,2$}.
\end{align*}

A rate distortion vector $(R_1,R_2,D_1,D_2)$ is strict-sense achievable if there exists a blocklength $n$, encoding functions $g_1^{(n)},g_2^{(n)}$ and a decoder $(\psi_1^{(n)},
\psi_2^{(n)})$ such that
\begin{align*}
R_i &\geq \frac{1}{n}\log M_i^{(n)} \mbox{~for $i=1,2$}\\
D_i &\geq  \mathbb{E} d({Y}_{i}^n,\hat{Y}_i^n) \mbox{~for $i=1,2$}.
\end{align*}
Where 
\begin{align*}
\hat{Y}_i^n=\psi_i^{(n)}(g_1^{(n)}(Y_1^n),g_2^{(n)}(Y_2^n)) \mbox{~for $i=1,2$}.
\end{align*}

\begin{definition}
Let $\mathcal{RD}^{\star}$ denote the set of strict-sense achievable rate distortion vectors and define the set of achievable rate distortion vectors to be its closure, $\overline{\mathcal{RD}}^{\star}$.
\end{definition}

Our ultimate goal in the present paper is to give a single-letter characterization of the region $\overline{\mathcal{RD}}^{\star}$.  However, in order to do this, we first consider an associated CEO problem.  In this sense, the roadmap for our argument is similar to that of \cite{bib:Wagner2008}.  Specifically, both arguments couple the multiterminal source coding problem to a parametrized family of CEO problems.  Then, the parameter in the CEO problem is ``tuned" to yield the converse result.  Despite this apparent similarity, the proofs are quite different since the results in \cite{bib:Wagner2008} depend heavily on the peculiarities of the Gaussian distribution.

\section{The CEO problem} \label{sec:ceoProblem}

In order to attack the general multiterminal problem, we begin by studying the CEO problem (See \cite{bib:BergerZhangViswanathan1996} for an introduction.).  To this end, let $\left\{ (X_j,Y_{1,j},Y_{2,j})\right\}_{j=1}^n=(X^n,Y_1^n,Y_2^n)$ be a sequence of $n$ independent, identically distributed random variables distributed according to the joint pmf $p(x,y_1,y_2)=p(x)p(y_1|x)p(y_2|x)$.  That is, $Y_1\leftrightarrow X \leftrightarrow Y_2$ form a Markov chain, in that order.

In this section, we consider the reproduction alphabet $\hat{\mathcal{X}}$ to be equal to the set of probability distributions over the source alphabet $\mathcal{X}$.  As before, for a vector $\hat{X}^n \in \hat{\mathcal{X}}^n$, we will use the notation $\hat{X}_j(x)$ to mean the $j^{th}$ coordinate of $\hat{X}^n$ (which is a probability distribution on $\mathcal{X}$) evaluated for the outcome $x\in\mathcal{X}$.  As in the rest of this paper, $d(\cdot,\cdot)$ is the logarithmic loss distortion measure.

A rate distortion CEO code (of blocklength $n$) consists of encoding functions:
\begin{align*}
g_i^{(n)}: \mathcal{Y}_i^n \rightarrow \left\{ 1,\dots,M_i^{(n)}\right\} \mbox{~for $i=1,2$}
\end{align*}
and a decoding function
\begin{align*}
\psi^{(n)}: \left\{ 1,\dots,M_1^{(n)} \right\} \times \left\{ 1,\dots,M_2^{(n)} \right\} \rightarrow \hat{\mathcal{X}}^n.
\end{align*}

A rate distortion vector $(R_1,R_2,D)$ is strict-sense achievable for the CEO problem if there exists a blocklength $n$, encoding functions $g_1^{(n)},g_2^{(n)}$ and a decoder $\psi^{(n)}$ such that
\begin{align*}
R_i &\geq \frac{1}{n}\log M_i^{(n)} \mbox{~for $i=1,2$}\\
D &\geq  \mathbb{E} d({X}^n,\hat{X}^n).
\end{align*}
Where 
\begin{align*}
\hat{X}^n=\psi^{(n)}(g_1^{(n)}(Y_1^n),g_2^{(n)}(Y_2^n)).
\end{align*}

\begin{definition}
Let $\mathcal{RD}^{\star}_{CEO}$ denote the set of strict-sense achievable rate distortion vectors and define the set of achievable rate distortion vectors to be its closure, $\overline{\mathcal{RD}}^{\star}_{CEO}$.
\end{definition}

\subsection{Inner Bound} \label{subsec:ceoInnerBound}

\begin{definition}
Let $(R_1,R_2,D)\in \mathcal{RD}_{CEO}^i$ if and only if there exists a joint distribution of the form 
\begin{align*}
p(x,y_1,y_2)p(u_1|y_1,q)p(u_2|y_2,q)p(q)
\end{align*}
where $|\mathcal{U}_1| \leq |\mathcal{Y}_1|$, $|\mathcal{U}_2| \leq |\mathcal{Y}_2|$, and $|\mathcal{Q}|\leq 4$, which satisfies
\begin{align*}
R_1 &\geq I(Y_1;U_1|U_2,Q)\\
R_2 &\geq I(Y_2;U_2|U_1,Q)\\
R_1 +R_2 &\geq I(U_1,U_2;Y_1,Y_2|Q)\\
D &\geq H(X|U_1,U_2,Q).
\end{align*}
%for some joint distribution $p(q)p(x,y_1,y_2)p(u_1|y_1,q)p(u_2|y_2,q)$.% with $|\mathcal{U}_j| \leq |\mathcal{Y}_j|+4$ and $|\mathcal{Q}|\leq 5$.
\end{definition}

\begin{theorem} \label{thm:CEOAchregion}
$\mathcal{RD}_{CEO}^i \subseteq \overline{\mathcal{RD}}^{\star}_{CEO}$.  That is, all rate distortion vectors $(R_1,R_2,D)\in \mathcal{RD}^i_{CEO}$ are achievable.
\end{theorem}

Before proceeding with the proof, we cite the following variant of a well-known inner bound:
\begin{proposition}[Berger-Tung Inner Bound \cite{bib:BergerLongo1977, bib:Tung1978}]\label{prop:BTCEO}
The rate distortion vector $(R_1,R_2,D)$ is achievable if  
\begin{align*}
R_1 &\geq I(U_1;Y_1|U_2,Q) \\
R_2 &\geq I(U_2;Y_2|U_1,Q) \\
R_1+R_2 &\geq I(U_1,U_2;Y_1,Y_2|Q) \\
D &\geq \mathbb{E}\left[d(X,f(U_1,U_2,Q) \right]% \\
%D_2 &\geq \mathbb{E}\left[d(Y_2,f_2(U_1,U_2,Q) \right].
\end{align*}
for a joint distribution 
\begin{align*}
%p(x,y_1,y_2,u_1,u_2,q) = 
p(x)p(y_1|x)p(y_2|x)p(u_1|y_1,q)p(u_2|y_2,q)p(q)
\end{align*}
and reproduction function
\begin{align*}
&f:\mathcal{U}_1\times \mathcal{U}_2\times \mathcal{Q}\rightarrow \hat{\mathcal{X}}.%\\
%&f_2:\mathcal{U}_1\times \mathcal{U}_2\times \mathcal{Q}\rightarrow \hat{\mathcal{Y}}_2.
\end{align*}
\end{proposition}
The proof of this proposition is a standard exercise in information theory, and is therefore omitted.  The interested reader is directed to the text  
\cite{bib:ElGamalYHKim2012} for a modern, detailed treatment.  The proposition follows from what is commonly called the Berger-Tung achievability scheme.  In this encoding scheme, each encoder quantizes its observation $Y_i^n$ to a codeword $U_i^n$, such that the empirical distribution of   the entries in $(Y_i^n,U_i^n)$ is very close to the true distribution $p(y_i,u_i)$.  In order to communicate their respective quantizations to the decoder, the encoders essentially perform Slepian-Wolf coding.  For this reason, the Berger-Tung achievability scheme is also referred to as a ``quantize-and-bin" coding scheme.

\begin{proof}[Proof of Theorem \ref{thm:CEOAchregion}]
%We first remark that the cardinality bounds on the alphabets in the definition of $\mathcal{RD}^{i}_{CEO}$ can be imposed without any loss of generality.  This is a consequence of \cite[Lemma 2.2]{bib:Jana2009} and is discussed in detail in Appendix \ref{app:cardBounds}.
%Next, g
Given Proposition \ref{prop:BTCEO}, the proof of Theorem \ref{thm:CEOAchregion} is immediate.  Indeed, if we apply Proposition \ref{prop:BTCEO} with the reproduction function $
f(U_1,U_2,Q)\triangleq\Pr\left[X=x| U_1,U_2,Q\right]$,  we note that 
\begin{align*}
\mathbb{E}\left[d(X,f(U_1,U_2,Q) \right]=H(X|U_1,U_2,Q),
\end{align*}
which yields the desired result.
\end{proof}

Thus, from the proof of Theorem \ref{thm:CEOAchregion}, we see that our inner bound  $\mathcal{RD}^{i}_{CEO}$ simply corresponds to a specialization of the general Berger-Tung inner bound to the case of logarithmic loss.

\subsection{A Matching Outer Bound} \label{subsec:ceoProblemOuterBound}
A particularly useful property of the logarithmic loss distortion measure is that the expected distortion is lower-bounded by a conditional entropy.  A similar property is enjoyed by Gaussian random variables under quadratic distortion.  In particular, if $G$ is Gaussian, and $\hat{G}$ is such that $\mathbb{E}(\hat{G}-G)^2\leq D$, then $\frac{1}{2}\log(2\pi e) D \geq h(G|\hat{G})$.  The case for  logarithmic loss  is similar, and we state it formally in the following lemma which is crucial in the proof of the converse.

\begin{lemma}\label{lem:minDistortion}
Let $Z=(g_1^{(n)}(Y_1^n),g_2^{(n)}(Y_2^n))$ be the argument of the reproduction function $\psi^{(n)}$.  Then $n\mathbb{E} d({X}^n,\hat{X}^n)\geq H(X^n|Z)$.
\end{lemma}
\begin{proof}
By definition of the reproduction alphabet, we can consider the reproduction $\hat{X}^n$ to be a probability distribution on $\mathcal{X}^n$ conditioned on the argument $Z$.  In particular, if $\hat{x}^n=\psi^{(n)}(z)$, define $s(x^n|z)\triangleq\prod_{j=1}^n \hat{x}_j(x_j)$.  It is readily verified that $s$ is a probability measure on $\mathcal{X}^n$.  Then, we obtain the following lower bound on the expected distortion conditioned on $Z=z$:
\begin{align*}
\mathbb{E}\left[ d({X}^n,\hat{X}^n) | Z=z \right] 
&=\frac{1}{n}\sum_{j=1}^n \sum_{x^n \in\mathcal{X}^n } p(x^n|z) \log\left(\frac{1}{\hat{x}_j(x_j)}\right)\\
&= \frac{1}{n} \sum_{x^n \in \mathcal{X}^n } p(x^n|z) \sum_{j=1}^n \log\left(\frac{1}{\hat{x}_j(x_j)}\right)\\
&= \frac{1}{n} \sum_{x^n \in \mathcal{X}^n } p(x^n|z) \log\left(\frac{1}{s(x^n|z)}\right)\\
&= \frac{1}{n} \sum_{x^n \in \mathcal{X}^n } p(x^n|z) \log\left(\frac{p(x^n|z)}{s(x^n|z)}\right) + \frac{1}{n}H(X^n|Z=z)\\
&= \frac{1}{n} D\left( p(x^n|z) \| s(x^n|z) \right) + \frac{1}{n}H(X^n|Z=z)\\
&\geq  \frac{1}{n}H(X^n|Z=z),
\end{align*}
where $p(x^n|z)=\Pr\left(X^n=x^n|Z=z \right)$ is the true conditional distribution.
Averaging both sides over all values of $Z$, we obtain the desired result.
\end{proof}

\begin{definition}
Let $(R_1,R_2,D) \in \mathcal{RD}_{CEO}^{o}$ if and only if there exists a joint distribution of the form 
\begin{align*}
%p(x,y_1,y_2,u_1,u_2,q) = 
p(x)p(y_1|x)p(y_2|x)p(u_1|y_1,q)p(u_2|y_2,q)p(q),
\end{align*}
which satisfies
\begin{align}
\left. \begin{array}{rl}
R_1 &\geq I(Y_1;U_1|X,Q)+H(X|U_2,Q)-D\\
R_2 &\geq I(Y_2;U_2|X,Q)+H(X|U_1,Q)-D\\
R_1 +R_2 &\geq I(U_1;Y_1|X,Q)+I(U_2;Y_2|X,Q)+H(X)-D\\
D &\geq H(X|U_1,U_2,Q).
\end{array} \right\} \label{eqn:RouterIneqs}
\end{align}
\end{definition}

\begin{theorem} \label{thm:CEOconv}
If $(R_1,R_2,D)$ is strict-sense achievable for the CEO problem, then $(R_1,R_2,D)\in\mathcal{RD}_{CEO}^{o}$.  
%That is, $\overline{\mathcal{RD}}^{\star}_{CEO} \subseteq \mathcal{RD}_{CEO}^o$.
\end{theorem}

\begin{proof}
Suppose the point $(R_1,R_2,D)$ is strict-sense achievable.  Let $A$ be a nonempty subset of $\{1,2\}$, and let $F_i=g_i^{(n)}(Y_i^n)$ be the message sent by encoder $i\in\{1,2\}$. Define $U_{i,j} \triangleq (F_i,Y_i^{j-1})$ and $Q_j\triangleq (X^{j-1},X_{j+1}^n)=X^n \backslash X_j$.  To simplify notation, let $Y_A=\cup_{i\in A} Y_i$ (similarly for $U_A$ and $F_A$). 

With these notations established, we have the following string of inequalities:
\begin{align}
n\sum_{i\in A} R_i &\geq \sum_{i\in A} H(F_i)\notag\\
&\geq H(F_A)\notag\\
&\geq I(Y_A^n;F_A|F_{A^c})\notag \\
&= I(X^n,Y_A^n;F_A|F_{A^c}) \label{eqn:fFuncY}\\
&= I(X^n;F_A|F_{A^c})+\sum_{i\in A} I(F_i;Y_i^n|X^n) \label{eqn:f_x_f} \\
&= H(X^n | F_{A^c})-H(X^n | F_1,F_2 )+ \sum_{i\in A} \sum_{j=1}^n I(Y_{i,j};F_i|X^n,Y_i^{j-1}) \notag \\
&\geq H(X^n | F_{A^c})+ \sum_{i\in A} \sum_{j=1}^n I(Y_{i,j};F_i|X^n,Y_i^{j-1}) -n D \label{eqn:distBnd}\\
&= \sum_{j=1}^n H(X_j | F_{A^c},X^{j-1}) + \sum_{i\in A} \sum_{j=1}^n I(Y_{i,j};F_i|X^n,Y_i^{j-1}) -n D \notag \\
&= \sum_{j=1}^n H(X_j | F_{A^c},X^{j-1}) + \sum_{i\in A} \sum_{j=1}^n I(Y_{i,j};U_{i,j}|X_j,Q_j) -n D \label{eqn:simpleMarkov} \\
&\geq \sum_{j=1}^n H(X_j | U_{{A^c},j},Q_j) + \sum_{i\in A} \sum_{j=1}^n I(Y_{i,j};U_{i,j}|X_j,Q_j) -n D \label{eqn:condRedEntropy}.
\end{align}
The nontrivial steps above can be justified as follows:
\begin{itemize}
\item \eqref{eqn:fFuncY} follows since $F_A$ is a function of $Y_A^n$.
\item \eqref{eqn:f_x_f} follows since $F_i$ is a function of $Y_i^n$ and $F_1\leftrightarrow X^n \leftrightarrow F_2$ form a Markov chain (since $Y_1^n \leftrightarrow X^n \leftrightarrow Y_2^n$ form a Markov chain).
\item \eqref{eqn:distBnd} follows since $nD \geq H(X^n|F_1,F_2)$ by Lemma \ref{lem:minDistortion}.
\item \eqref{eqn:simpleMarkov} follows from the Markov chain $Y_{i,j}\leftrightarrow X^n \leftrightarrow Y_i^{j-1}$, which follows from the i.i.d. nature of the source sequences.
\item \eqref{eqn:condRedEntropy} simply follows from the fact that conditioning reduces entropy.
\end{itemize}
Therefore, dividing both sides by $n$, we have:
\begin{align*}
\sum_{i\in A} R_i &\geq \frac{1}{n}\sum_{j=1}^n H(X_j | U_{A^c,j},Q_j)+ \sum_{i\in A} \frac{1}{n}\sum_{j=1}^n I(Y_{i,j};U_{i,j}|X_j,Q_j) - D.
\end{align*}
Also, using Lemma \ref{lem:minDistortion} and the fact that conditioning reduces entropy:
\begin{align*}
D \geq \frac{1}{n}H(X^n|F_1,F_2)\geq \frac{1}{n}\sum_{j=1}^n H(X_j|U_{1,j},U_{2,j},Q_j).
\end{align*}
Observe that $Q_j$ is independent of $(X_j,Y_{1,j},Y_{2,j})$ and, conditioned on $Q_j$, we have the long Markov chain $U_{1,j}\leftrightarrow Y_{1,j}\leftrightarrow X_j\leftrightarrow Y_{2,j}\leftrightarrow U_{2,j}$.  Finally, by a standard time-sharing argument, we conclude by saying that if $(R_1,R_2,D)$ is strict-sense achievable for the CEO problem, then 
\begin{align*}
R_1 &\geq I(Y_1;U_1|X,Q)+H(X|U_2,Q)-D\\
R_2 &\geq I(Y_2;U_2|X,Q)+H(X|U_1,Q)-D\\
R_1 +R_2 &\geq I(U_1;Y_1|X,Q)+I(U_2;Y_2|X,Q)+H(X)-D\\
D &\geq H(X|U_1,U_2,Q).
\end{align*}
for some joint distribution $p(q)p(x,y_1,y_2)p(u_1|y_1,q)p(u_2|y_2,q)$.
\end{proof}

\begin{theorem}\label{thm:CEOregion}
$\mathcal{RD}^{o}_{CEO} = \mathcal{RD}^{i}_{CEO}=\overline{\mathcal{RD}}^{\star}_{CEO}$. 
\end{theorem}
\begin{proof}
%We first remark that the cardinality bounds in the definition of $\mathcal{RD}^{i}_{CEO}$ can be imposed without any loss of generality.  This is a consequence of \cite[Lemma 2.2]{bib:Jana2009} and is discussed in detail in the full manuscript \cite{bib:CourtadeWeissman2011}.
%
%
We first remark that the cardinality bounds on the alphabets in the definition of $\mathcal{RD}^{i}_{CEO}$ can be imposed without any loss of generality.  This is a consequence of \cite[Lemma 2.2]{bib:Jana2009} and is discussed in detail in Appendix \ref{app:cardBounds}.

Therefore, it will suffice to show $\mathcal{RD}^{o}_{CEO} \subseteq \mathcal{RD}^{i}_{CEO}$ without considering the cardinality bounds.
To this end, fix $p(q)$, $p(u_1|y_1,q)$, and $p(u_2|y_2,q)$ and consider the extreme points\footnote{For two encoders, it is easy enough to enumerate the extreme points by inspection.  However, this can be formalized by a submodularity argument, which is given in Appendix \ref{app:mEncCEO}.} of polytope defined by the inequalities \eqref{eqn:RouterIneqs}:
\begin{align*}
P_1&=\bigg(0,0,I(Y_1;U_1|X,Q)+I(Y_2;U_2|X,Q)+H(X)\bigg) \\
P_2&=\bigg(I(Y_1;U_1|Q),0,I(U_2;Y_2|X,Q)+H(X|U_1,Q)\bigg)\\
P_3&=\bigg(0,I(Y_2;U_2|Q),I(U_1;Y_1|X,Q)+H(X|U_2,Q)\bigg)\\
P_4&=\bigg(I(Y_1;U_1|Q),I(Y_2;U_2|U_1,Q),H(X|U_1,U_2,Q)\bigg)\\
P_5&=\bigg(I(Y_1;U_1|U_2,Q),I(Y_2;U_2|Q),H(X|U_1,U_2,Q)\bigg),
\end{align*}
where the point $P_j$ is a triple $(R_1^{(j)},R_2^{(j)},D^{(j)})$.  We say a point $(R_1^{(j)},R_2^{(j)},D^{(j)})$ is \emph{dominated by} a point in $\mathcal{RD}_{CEO}^i$ if there exists some $(R_1,R_2,D)\in\mathcal{RD}_{CEO}^i$ for which $R_1\leq R_1^{(j)}$, $R_2\leq R_2^{(j)}$, and $D\leq D^{(j)}$. Observe that each of the extreme points $P_1,\dots,P_5$ is dominated by a point in $\mathcal{RD}_{CEO}^i$:  
\begin{itemize}
\item First, observe that $P_4$ and $P_5$ are both in $\mathcal{RD}_{CEO}^i$, so these points are not problematic.  
\item Next, observe that the point $(0,0,H(X))$ is in $\mathcal{RD}_{CEO}^i$, which can be seen by setting all auxiliary random variables to be constant.  This point dominates $P_1$.  
\item By using auxiliary random variables $(\hat{U}_1,\hat{U}_2,Q)=(U_1,\emptyset,Q)$, the point 
 $(I(Y_1;U_1|Q),0,H(X|U_1,Q))$ is in $\mathcal{RD}^i_{CEO}$, and dominates the point  $P_2$.  By a symmetric argument, the point $P_3$ is also dominated by a point in $\mathcal{RD}^i_{CEO}$.
\end{itemize}
Since $\mathcal{RD}_{CEO}^o$ is the convex hull of all such extreme points (i.e., the convex hull of the union of extreme points over all appropriate joint distributions), the theorem is proved.
\end{proof}

\begin{remark}
Theorem \ref{thm:CEOregion} can be extended to the general case of  $m$-encoders.  Details are provided in Appendix \ref{app:mEncCEO}.
\end{remark}

\subsection{A stronger converse result for the CEO problem} \label{subsec:ceoStrongerConverse}
As defined, our reproduction sequence $\hat{X}^n$ is restricted to be a product distribution on $\mathcal{X}^n$.  However, for a blocklength $n$ code, we can allow $\hat{X}^n$ to be \emph{any} probability distribution on $\mathcal{X}^n$ and the converse result still holds.  In this case, we define the sequence distortion as follows:
\begin{align*}
d({x}^n,\hat{x}^n)=\frac{1}{n}\log\left(\frac{1}{\hat{x}^n(x^n)} \right),
\end{align*}
which is compatible with the original definition when $\hat{X}^n$ is a product distribution. The reader can verify that the result of Lemma \ref{lem:minDistortion} is still true for this more general distortion alphabet by setting $s(x^n|z)=\hat{x}^n(x^n)$ in the corresponding proof.  Since Lemma \ref{lem:minDistortion} is the key tool in the CEO converse result, this implies that the converse holds even if $\hat{X}^n$ is allowed to be any probability distribution on $\mathcal{X}^n$ (rather than being restricted to the set of product distributions).

When this stronger converse result is taken together with the achievability result, we observe that restricting $\hat{X}^n$ to be a product distribution is in fact optimal and can achieve all points in $\overline{\mathcal{RD}}^{\star}_{CEO}$.

\subsection{An Example: Distributed compression of a posterior distribution} \label{subsec:ceoExamples}

Suppose two sensors observe sequences $Y_1^n$ and $Y_2^n$ respectively, which are conditionally independent given a hidden sequence $X^n$.  The sensors communicate with a fusion center through rate-limited links of capacity $R_1$ and $R_2$ respectively.  Given sequences $(Y_1^n,Y_2^n)$ are observed, the sequence $X^n$ cannot be determined in general, so the fusion center would like to estimate the posterior distribution $p(x^n|Y_1^n,Y_2^n)$.  Since the communication links are rate-limited, the fusion center cannot necessarily compute $p(x^n|Y_1^n,Y_2^n)$ exactly.  In this case, the fusion center would like to generate an estimate $\hat{p}(x^n|g_1^{(n)}(Y_1^n), g_2^{(n)}(Y_2^n) )$ that should approximate $p(x^n|Y_1^n,Y_2^n)$ in the sense that, on average:
\begin{align*}
D\Big(p(x^n|y_1^n,y_2^n) \Big\| \hat{p}(x^n|g_1^{(n)}(y_1^n), g_2^{(n)}(y_2^n) ) \Big)\leq n\varepsilon,
\end{align*}
where, consistent with standard notation (e.g. \cite{bib:CoverThomas2006}), we write \linebreak $D (p(x^n|y_1^n, y_2^n) \| \hat{p}(x^n|g_1^{(n)}(y_1^n), g_2^{(n)}(y_2^n) ))$ as shorthand for 
\begin{align*}
%&\mathbb{E}_{p(x^n,y_1^n,y_2^n)} \log \frac{p(x^n|Y_1^n, Y_2^n)}{\hat{p}(x^n|g_1^{(n)}(Y_1^n), g_2^{(n)}(Y_2^n) )} \\
%&\quad = 
\sum_{x^n, y_1^n, y_2^n} p(x^n, y_1^n, y_2^n) \log \frac{p(x^n|y_1^n, y_2^n)}{\hat{p}(x^n|g_1^{(n)}(y_1^n), g_2^{(n)}(y_2^n) )}.
\end{align*} 
The relevant question here is the following. What is the minimum distortion $\varepsilon$ that is attainable given $R_1$ and $R_2$?

Considering the CEO problem for this setup, we have:
\begin{align*}
\mathbb{E}d(\hat{X}^n,X^n) &= \frac{1}{n}\sum_{(x^n,y_1^n,y_2^n)} p(x^n,y_1^n,y_2^n) 
\log \left( \frac{1}{\hat{x}^n(x^n)}\right)\\
&=\frac{1}{n} D\Big(p(x^n|y_1^n,y_2^n) \Big\| \hat{x}^n(x^n) \Big) + \frac{1}{n}H(X^n|Y_1^n,Y_2^n).
\end{align*}
Identifying $\hat{p}(x^n|g_1^{(n)}(Y_1^n), g_2^{(n)}(Y_2^n) )\leftarrow \hat{X}^n(x^n)$, we have:
\begin{align*}
D\Big(p(x^n|y_1^n,y_2^n) \Big\|  \hat{p}(x^n|g_1^{(n)}(y_1^n), g_2^{(n)}(y_2^n) ) \Big) = n\mathbb{E}d(\hat{X}^n,X^n)-n H(X|Y_1,Y_2).
\end{align*}
Thus, finding the minimum possible distortion reduces to an optimization problem over $\overline{\mathcal{RD}}_{CEO}^{\star} $.  In particular, the minimum attainable distortion $\varepsilon^*$ is given by
\begin{align}
\varepsilon^* = \inf \left\{D : (R_1,R_2,D)\in \overline{\mathcal{RD}}_{CEO}^{\star} \right\} -H(X|Y_1,Y_2). \label{eqn:minKL}
\end{align}
Moreover, the minimum distortion is obtained by estimating each $x_j$ separately.  In other words, there exists an optimal (essentially, for large $n$) estimate $\hat{p}^*(x^n|\cdot, \cdot )$ (which is itself a function of optimal encoding functions $g_1^{*(n)}(\cdot)$ and $g_2^{*(n)}(\cdot)$) that can be expressed as a product distribution
\begin{align*}
\hat{p}^*(x^n|\cdot, \cdot )=\prod_{j=1}^n \hat{p}_j^*\left(x_j|g_1^{*(n)}(\cdot), g_2^{*(n)}(\cdot) \right).
\end{align*}  
For this choice of $\hat{p}^*(x^n|\cdot, \cdot )$, we have the following relationship:
\begin{align*}
%\frac{1}{n}D\left(p(x^n|y_1^n,y_2^n) \| q^*(x^n) \right)= \frac{1}{n}\sum_{j=1}^n D\left(p(x(j)|y_1(j),y_2(j)) \| q_j^*(x(j)) \right) =\varepsilon^*.
\frac{1}{n}\sum_{j=1}^n D\Big(p(x_j|y_{1,j},y_{2,j}) \Big\| \hat{p}_j^*\left(x_j|g_1^{*(n)}(y_1^n), g_2^{*(n)}(y_2^n) \right) \Big) =\varepsilon^*.
\end{align*}
In light of this fact, one can apply Markov's inequality to obtain the following estimate on peak component-wise distortion:
\begin{align*}
\# \Bigg\{ j ~\Big| ~D\Big(p(x_j|y_{1,j},y_{2,j}) \Big\| \hat{p}_j^*\left(x_j|g_1^{*(n)}(y_1^n), g_2^{*(n)}(y_2^n) \right) \Big) \geq \zeta \Bigg\} \leq n\frac{\varepsilon^*}{\zeta},
\end{align*}
where $\#(\cdot)$ is the counting measure.

To make this example more concrete, consider the scenario depicted in Figure \ref{fig:BSCex}, where $X\sim \mbox{Bernoulli}(\frac{1}{2})$ and $Y_i$ is the result of passing $X$ through a binary symmetric channel with crossover probability $\alpha$ for $i=1,2$.  To simplify things, we constrain the rates of each encoder to be at most $R$ bits per channel use.

\begin{figure}
\begin{center}
\def\svgwidth{6in}
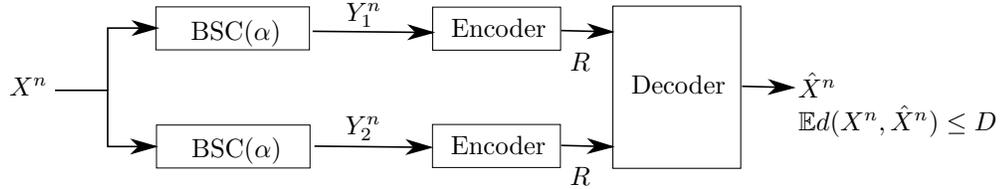
\caption[A symmetric CEO source coding network.]{An example CEO problem where $X\sim \mbox{Bernoulli}(\frac{1}{2})$, $\Pr(Y_i=X)=(1-\alpha)$, and both encoders are subject to the same rate constraint.}\label{fig:BSCex}
\end{center}
\end{figure}

By performing a brute-force search over a fine mesh of conditional distributions  $\{p(u_i|y_i)\}_{i=1}^2$, we numerically approximate the set of $(R,D)$ pairs such that $(R,R,D)$ is in the achievable region $\overline{\mathcal{RD}}_{CEO}^{\star}$ corresponding to the network in Figure \ref{fig:BSCex}.  The lower convex envelope of these $(R,D)$ pairs is plotted in Figure \ref{fig:distRateFn} for $\alpha\in\{0.01,0.1,0.25\}$.  Continuing our example above for this concrete choice of source parameters, we compute the minimum achievable Kullback-Leibler distance $\varepsilon^*$ according to \eqref{eqn:minKL}.  The result is given in Figure \ref{fig:epsRateFn}. 

These numerical results are intuitively satisfying in the sense that, if $Y_1,Y_2$ are high-quality estimates of $X$ (e.g., $\alpha=0.01$), then a small increase in the allowable rate $R$ results in a large relative improvement of $\hat{p}(x|\cdot,\cdot)$, the decoder's estimate of $p(x|Y_1,Y_2)$.  On the other hand, if $Y_1,Y_2$ are poor-quality estimates of $X$ (e.g., $\alpha=0.25$), then we require a large increase in the allowable rate $R$ in order to obtain an appreciable improvement of $\hat{p}(x|\cdot,\cdot)$.

% \hspace{-10pt}
\begin{figure}[t]
  \centering
  \includegraphics[trim = 37mm 82mm 37mm 87mm, clip, scale=.75]{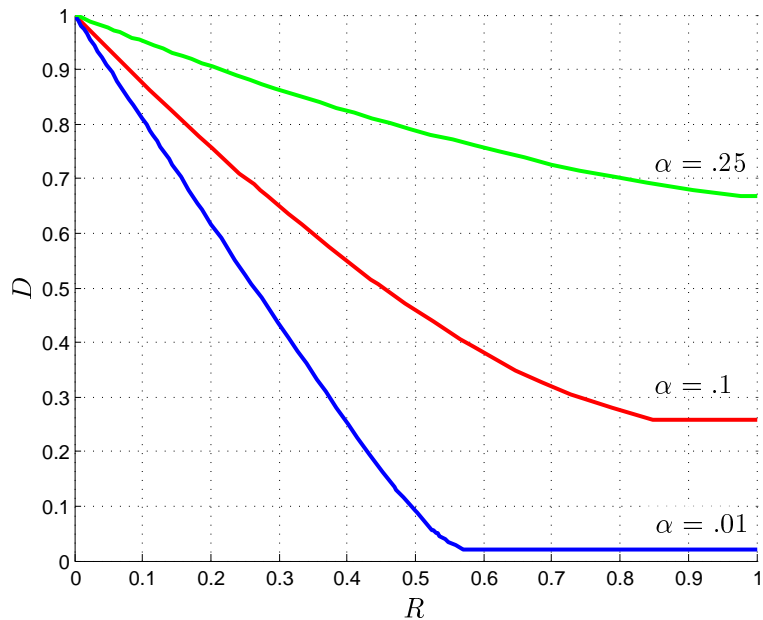}
  \caption[Distortion-rate function for the symmetric CEO   network.]{The distortion-rate function of the network in Figure \ref{fig:BSCex} computed for $\alpha\in\{0.01,0.1,0.25\}$.} \label{fig:distRateFn}         
  \end{figure}
  
\begin{figure}[t]  
\centering     
  \includegraphics[trim = 37mm 82mm 37mm 87mm, clip, scale=.75]{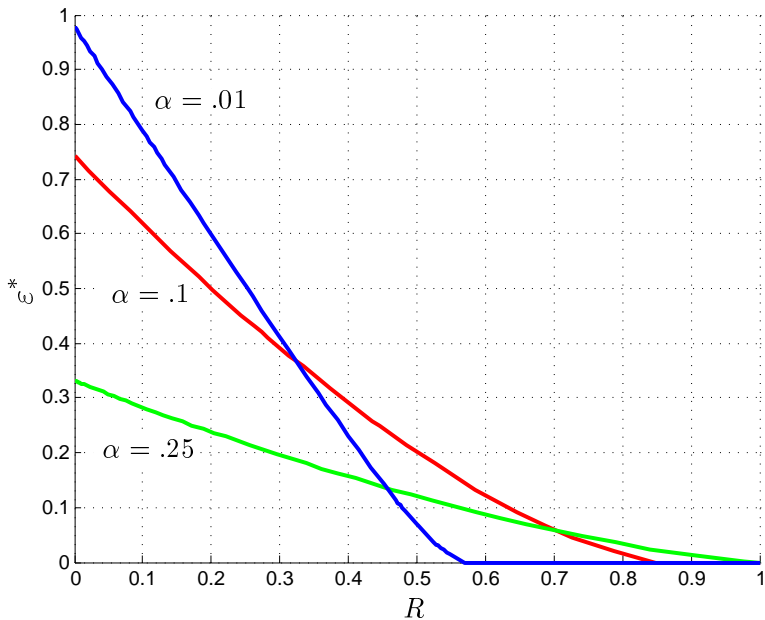}
  \caption[Achievable KL divergence for the symmetric CEO network.]{The minimum achievable Kullback-Leibler distance computed according to \eqref{eqn:minKL}, i.e., the curves here are those of Figure \ref{fig:distRateFn}, lowered by the constant $H(X|Y_1, Y_2)$.}\label{fig:epsRateFn}
\end{figure}
%\begin{figure}
%  %\centering
%  \hspace{-30pt}
%  \subfloat[The distortion-rate function.]{\label{fig:distRateFn}\includegraphics[trim = 37mm 82mm 37mm 87mm, clip, scale=.5]{rd_symmetricProb.pdf}}                
%  \subfloat[Minimum Kullback-Leibler distance.]{\label{fig:epsRateFn}\includegraphics[trim = 37mm 82mm 37mm 87mm, clip, scale=.5]{rd_symmetricProbEps.pdf}}
%  \caption{In (a), we plot the distortion-rate function of the network in Figure \ref{fig:BSCex} computed for $\alpha\in\{0.01,0.1,0.25\}$.  Subfigure (b) gives the minimum achievable Kullback-Leibler distance computed according to \eqref{eqn:minKL}, i.e., the curves in (b) are those of (a), lowered by the constant $H(X|Y_1, Y_2)$.}
%  \label{fig:numResults}
%\end{figure}

One field where this example is directly applicable is machine learning.  In this case, $X_j$ could represent the class of object $j$, and $Y_{1,j},Y_{2,j}$ are observable attributes.  In machine learning, one typically estimates the probability that an object belongs to a particular class given a set of observable attributes.  For this type of estimation problem, relative entropy is a natural penalty criterion.

Another application is to horse-racing with conditionally independent, rate-limited side informations.  In this case, the doubling rate of the gambler's wealth can be expressed in terms of the logarithmic loss distortion measure.  This example is consistent with the original interpretation of the CEO problem, where the CEO makes consecutive business decisions (investments) having outcomes $X^n$, with the objective of maximizing the wealth of the company. We omit the details.

\subsection{An Example: Joint estimation of the encoder observations}

Suppose one wishes to estimate the encoder observations $(Y_1,Y_2)$.  In this case, the rate region simplifies considerably. In particular, if we tolerate a distortion $D$ in our estimate of the pair $(Y_1,Y_2)$, then the achievable rate region is the same as the Slepian-Wolf rate region with each rate constraint relaxed by $D$ bits. Formally:

\begin{theorem} \label{thm:SWminusD}
If $X=(Y_1,Y_2)$, then $\overline{\mathcal{RD}}_{CEO}^{\star}$ consists of all vectors $(R_1,R_2,D)$ satisfying
\begin{align*}
R_1 &\geq H(Y_1 | Y_2)-D \\
R_2 &\geq H(Y_2 | Y_1)-D \\
R_1+R_2 &\geq H(Y_1,Y_2)-D\\
D &\geq 0.
\end{align*}
\end{theorem}
\begin{proof}
First, note that Theorem  \ref{thm:CEOregion} implies that $\overline{\mathcal{RD}}_{CEO}^{\star}$ is equivalent to the the union of $(R_1,R_2,D)$ triples satisfying \eqref{eqn:RouterIneqs} taken over all joint distributions $p(q)p(x,y_1,y_2)p(u_1|y_1,q)p(u_2|y_2,q)$.   Now, since $X=(Y_1,Y_2)$, each of the inequalities \eqref{eqn:RouterIneqs} can be lower bounded as follows:
\begin{align*}
R_1 &\geq I(Y_1;U_1|Y_1,Y_2,Q)+H(Y_1,Y_2|U_2,Q)-D\\
&=  H(Y_2|U_2,Q)+H(Y_1|Y_2)-D\\
&\geq H(Y_1|Y_2)-D\\
R_2 &\geq I(Y_2;U_2|Y_1,Y_2,Q)+H(Y_1,Y_2|U_1,Q)-D\\
&=  H(Y_1|U_1,Q)+H(Y_2|Y_1)-D\\
&\geq H(Y_2|Y_1)-D\\
R_1 +R_2 &\geq I(U_1;Y_1|Y_1,Y_2,Q)+I(U_2;Y_2|Y_1,Y_2,Q)+H(Y_1,Y_2)-D\\
&=H(Y_1,Y_2)-D\\
D &\geq H(Y_1,Y_2|U_1,U_2,Q)\\
&\geq 0.
\end{align*}
Finally, observe that by setting $U_i=Y_i$ for $i=1,2$, we can achieve any point in this relaxed region (again, a consequence of Theorem \ref{thm:CEOregion}).
\end{proof}

We remark that this result was first proved in \cite{bib:CourtadeISIT2011} by Courtade and Wesel using a different method.

\section{Multiterminal Source Coding} \label{sec:MTSC}
With Theorem \ref{thm:CEOregion} in hand, we are now in a position to characterize the achievable rate distortion region $\overline{\mathcal{RD}}^{\star}$ for the multiterminal source coding problem under logarithmic loss.  As before, we prove an inner bound first.

%\subsection{Inner Bound} \label{subsec:MTSCinnerBound}

\subsection{Inner Bound} \label{subsec:MTSCinnerBound}
\begin{definition}
Let $(R_1,R_2,D_1,D_2) \in \mathcal{RD}^{i}$ if and only if there exists a joint distribution of the form 
\begin{align*}
%p(x,y_1,y_2,u_1,u_2,q) = 
p(y_1,y_2)p(u_1|y_1,q)p(u_2|y_2,q)p(q)
\end{align*}
where $|\mathcal{U}_1|\leq |\mathcal{Y}_1|$, $|\mathcal{U}_2|\leq |\mathcal{Y}_2|$,  and $|\mathcal{Q}|\leq 5$, which satisfies
\begin{align*}
R_1 &\geq I(Y_1;U_1|U_2,Q)\\
R_2 &\geq I(Y_2;U_2|U_1,Q)\\
R_1 +R_2 &\geq I(U_1,U_2;Y_1,Y_2|Q)\\
D_1 &\geq H(Y_1|U_1,U_2,Q)\\
D_2 &\geq H(Y_2|U_1,U_2,Q).
\end{align*}
\end{definition}

\begin{theorem} \label{thm:MTSCachv}
$\mathcal{RD}^i\subseteq \overline{\mathcal{RD}}^{\star}$.  That is, all rate distortion vectors in $\mathcal{RD}^i$ are achievable.
\end{theorem}
%\begin{proof}
%Use the standard Berger-Tung achievability scheme (originally described in \cite{bib:tung}, see \cite{bib:ElGamal} for a detailed treatment).  It is readily verified (via a strong typicality exercise) that setting
%\begin{align*}
%\hat{X}[x](j)=\Pr\left[X(j)=x | U_1(j),U_2(j),Q(j)\right]
%\end{align*}
%achieves expected distortion:
%\begin{align*}
%\frac{1}{n}\sum_{j=1}^n d(\hat{X}(j),X(j)) \leq H(X|U_1,U_2,Q)+\epsilon_n,
%\end{align*}
%where $\lim_{n\rightarrow\infty}\epsilon_n=0$.
%\end{proof}
%
Again, we require an appropriate version of the Berger-Tung inner bound:
\begin{proposition}[Berger-Tung Inner Bound \cite{bib:BergerLongo1977, bib:Tung1978}]\label{prop:BTMTSC}
The rate distortion vector $(R_1,R_2,D_1,D_2)$ is achievable if  
\begin{align*}
R_1 &\geq I(U_1;Y_1|U_2,Q) \\
R_2 &\geq I(U_2;Y_2|U_1,Q) \\
R_1+R_2 &\geq I(U_1,U_2;Y_1,Y_2|Q) \\
D_1 &\geq \mathbb{E}\left[d(Y_1,f_1(U_1,U_2,Q) \right]\\
D_2 &\geq \mathbb{E}\left[d(Y_2,f_2(U_1,U_2,Q) \right].% \\
%D_2 &\geq \mathbb{E}\left[d(Y_2,f_2(U_1,U_2,Q) \right].
\end{align*}
for a joint distribution 
\begin{align*}
%p(x,y_1,y_2,u_1,u_2,q) = 
p(y_1,y_2)p(u_1|y_1,q)p(u_2|y_2,q)p(q)
\end{align*}
and reproduction functions
\begin{align*}
&f_i:\mathcal{U}_1\times \mathcal{U}_2\times \mathcal{Q}\rightarrow \hat{\mathcal{Y}}_i, \mbox{~for $i=1,2$.}
%&f_2:\mathcal{U}_1\times \mathcal{U}_2\times \mathcal{Q}\rightarrow \hat{\mathcal{Y}}_2 %\\
%&f_2:\mathcal{U}_1\times \mathcal{U}_2\times \mathcal{Q}\rightarrow \hat{\mathcal{Y}}_2.
\end{align*}
\end{proposition}
%\begin{proposition}[Berger-Tung Inner Bound \cite{bib:BergerLongo1977}]
%The rate distortion vector $(R_1,R_2,D_1,D_2)$ is achievable if  
%\begin{align*}
%R_1 &\geq I(U_1;Y_1|U_2,Q) \\
%R_2 &\geq I(U_2;Y_2|U_1,Q) \\
%R_1+R_2 &\geq I(U_1,U_2;Y_1,Y_2|Q) \\
%D_1 &\geq \mathbb{E}\left[d(Y_1,f_1(U_1,U_2,Q) \right] \\
%D_2 &\geq \mathbb{E}\left[d(Y_2,f_2(U_1,U_2,Q) \right].
%\end{align*}
%for some distribution 
%\begin{align*}
%p(y_1,y_2,u_1,u_2,q) = p(y_1,y_2)p(u_1|y_1,q)p(u_2|y_2,q)p(q)
%\end{align*}
%and functions 
%\begin{align*}
%&f_1:\mathcal{U}_1\times \mathcal{U}_2\times \mathcal{Q}\rightarrow \hat{\mathcal{Y}}_1\\
%&f_2:\mathcal{U}_1\times \mathcal{U}_2\times \mathcal{Q}\rightarrow \hat{\mathcal{Y}}_2.
%\end{align*}
%\end{proposition}

\begin{proof}[Proof of Theorem \ref{thm:MTSCachv}]
%As in the proof of Theorem \ref{thm:CEOAchregion}, we note that the cardinality bounds on the alphabets in the definition of $\mathcal{RD}^{i}$ can be imposed without any loss of generality.  This is discussed in detail in Appendix \ref{app:cardBounds}.

To prove the theorem, we simply apply Proposition \ref{prop:BTMTSC} with the reproduction functions $
f_i(U_1,U_2,Q):=\Pr\left[Y_i=y_i| U_1,U_2,Q\right]$.\end{proof}

Hence, we again see that our inner bound $\mathcal{RD}^i\subseteq \overline{\mathcal{RD}}^{\star}$ is nothing more than the Berger-Tung inner bound specialized to the setting when distortion is measured under logarithmic loss.

\subsection{A Matching Outer Bound} \label{subsec:MTSouterBound}

The main result of this paper is the following theorem.
\begin{theorem}\label{thm:MTSCregionDesc}
$\mathcal{RD}^i=\overline{\mathcal{RD}}^{\star}$.
\end{theorem}

\begin{proof}
As before, we note that the cardinality bounds on the alphabets in the definition of $\mathcal{RD}^{i}$ can be imposed without any loss of generality.  This is discussed in detail in Appendix \ref{app:cardBounds}.

Assume $(R_1,R_2,D_1,D_2)$ is strict-sense achievable.  Observe that proving that $(R_1,R_2,D_1,D_2)\in \mathcal{RD}^i$ will prove the theorem, since $\mathcal{RD}^i\subseteq \overline{\mathcal{RD}}^{\star}$ and $\overline{\mathcal{RD}}^{\star}$ is closed by definition.

For convenience, define $\mathcal{P}(R_1,R_2)$ to be the set of joint distributions of the form 
\begin{align*}
p(y_1,y_2)p(u_1|y_1,q)p(u_2|y_2,q)p(q)
\end{align*}
with $|\mathcal{U}_1|\leq |\mathcal{Y}_1|$, $|\mathcal{U}_2|\leq |\mathcal{Y}_2|$, and $|\mathcal{Q}|\leq 4$ satisfying 
\begin{align*}
R_1 &\geq I(U_1;Y_1|U_2,Q) \\
R_2 &\geq I(U_2;Y_2|U_1,Q)\\
R_1+R_2 &\geq I(U_1,U_2;Y_1,Y_2|Q).
\end{align*}
We remark that $\mathcal{P}(R_1,R_2)$ is compact.  We also note that it will suffice to show the existence of a joint distribution in $\mathcal{P}(R_1,R_2)$ satisfying $H(Y_1|U_1,U_2,Q)\leq D_1$ and $H(Y_2|U_1,U_2,Q)\leq D_2$ to prove that $(R_1,R_2,D_1,D_2)\in \mathcal{RD}^i$.

With foresight, consider random variable $X$ defined as follows
\begin{align}
X=\left\{ \begin{array}{ll}
(Y_1,1) & \mbox{with probability $t$}\\
(Y_2,2) & \mbox{with probability $1-t$.}
\end{array} \right. \label{eqn:XassignedY1Y2}
\end{align}
In other words, $X=(Y_B,B)$, where $B$ is a Bernoulli random variable independent of $Y_1,Y_2$.  Observe that $Y_1 \leftrightarrow X \leftrightarrow Y_2$ form a Markov chain, and thus, we are able to apply Theorem \ref{thm:CEOregion}.  

Since $(R_1,R_2,D_1,D_2)$ is strict-sense achievable, the decoder can construct reproductions $\hat{Y}_1^n,\hat{Y}_2^n$ satisfying 
\begin{align*}
\frac{1}{n}\sum_{j=1}^n \mathbb{E}d({Y}_{i,j},\hat{Y}_{i,j})\leq D_i \mbox{~for $i=1,2$.}
\end{align*}
Fix the encoding operations and set $\hat{X}_j\left((y_1,1)\right)=t\hat{Y}_{1,j}(y_1)$ and $\hat{X}_j\left((y_2,2)\right)=(1-t)\hat{Y}_{2,j}(y_2)$.  Then for the CEO problem defined by $(X,Y_1,Y_2)$:
\begin{align}
&\frac{1}{n}\sum_{j=1}^n \mathbb{E}d({X}_j,\hat{X}_j) \notag \\
&= \frac{t}{n}\sum_{j=1}^n \mathbb{E} \log\left(\frac{1}{t\hat{Y}_{1,j}(Y_{1,j})} \right)  +\frac{1-t}{n}\sum_{j=1}^n \mathbb{E} \log\left(\frac{1}{(1-t)\hat{Y}_{2,j}(Y_{2,j})} \right)\notag  \\
&=h_2(t)+\frac{t}{n}\sum_{j=1}^n \mathbb{E}d({Y}_{1,j},\hat{Y}_{1,j}) +\frac{1-t}{n}\sum_{j=1}^n \mathbb{E}d({Y}_{2,j},\hat{Y}_{2,j})\notag \\
&\leq h_2(t)+tD_1+(1-t)D_2\notag% \label{eqn:ceoCoupling}
\end{align}
where $h_2(t)$ is the binary entropy function.  Hence, for this CEO problem, distortion $h_2(t)+tD_1+(1-t)D_2$ is achievable and Theorem \ref{thm:CEOregion} yields a joint distribution\footnote{Henceforth, we use the superscript $(t)$ to explicitly denote the dependence of the auxiliary random variables on the distribution parametrized by $t$.} $P_t \in \mathcal{P}(R_1,R_2)$ satisfying 
\begin{align*}
h_2(t)+tD_1+(1-t)D_2 &\geq H(X|U_1^{(t)},U_2^{(t)},Q^{(t)})\\
&=h_2(t)+tH(Y_1|U_1^{(t)},U_2^{(t)},Q^{(t)})\\
&\quad + (1-t)H(Y_2|U_1^{(t)},U_2^{(t)},Q^{(t)}), 
\end{align*}
where the second equality follows by 
by definition of $X$ in \eqref{eqn:XassignedY1Y2}.  For convenience, define $H_1(P_t) \triangleq H(Y_1|U_1^{(t)},U_2^{(t)},Q^{(t)})$ and $H_2(P_t) \triangleq H(Y_2|U_1^{(t)},U_2^{(t)},Q^{(t)})$.  Note the following two facts:
\begin{enumerate}
\item By continuity of entropy, the functions $H_1(\cdot)$ and $H_2(\cdot)$ are continuous on the compact domain $\mathcal{P}(R_1,R_2)$.
\item The above argument proves the existence of a function $\varphi : [0,1] \rightarrow \mathcal{P}(R_1,R_2)$ which satisfies
\begin{align*}
t H_1(\varphi(t)) + (1-t)H_2(\varphi(t)) \leq tD_1 + (1-t)D_2 \mbox{~ for all $t\in[0,1]$.}
\end{align*}
\end{enumerate}

These two facts satisfy the requirements of Lemma \ref{lem:amplifyConvex} (see Appendix \ref{app:convexityLemma}), and hence there exists $P_{t_1}\in\mathcal{P}(R_1,R_2)$, $P_{t_2}\in\mathcal{P}(R_1,R_2)$, and $\theta\in[0,1]$ for which
\begin{align*}
\theta H_1(P_{t_1}) + (1-\theta) H_1(P_{t_2}) &\leq D_1 \\
\theta H_2(P_{t_1}) + (1-\theta) H_2(P_{t_2}) &\leq D_2.
\end{align*}

Timesharing\footnote{The timesharing scheme can be embedded in the timesharing variable $Q$, increasing the cardinality of $\mathcal{Q}$ by a factor of two.} between distributions $P_{t_1}$ and $P_{t_2}$ with probabilities $\theta$ and $(1-\theta)$, respectively, yields a distribution $P^*\in \mathcal{P}(R_1,R_2)$ which satisfies $H_1(P^*)\leq D_1$ and $ H_2(P^*) \leq D_2$.
This proves the theorem.
\end{proof}

%\begin{figure}[h]
%\begin{center}
%%\def\svgwidth{6in}
%%\input{figs/f1_f2.pdf_tex}
%\includegraphics[trim = 30mm 73mm 30mm 90mm, clip, scale=.7]{rd_pigeonhole_niceFig.pdf}
%\caption{An example of piecewise-linear functions $f_1(t)$ and $f_2(t)$ that satisfy $tf_1(t)+(1-t)f_2(t)\leq tD_1 + (1-t)D_2$ whenever $t\in\{0,1/10,\dots,9/10, 1\}$.  The vertical bar indicates the location of a $t^*$ where $f_1(t^*)\leq D_1$ and $f_2(t^*)\leq D_2$.}\label{fig:f1_f2}
%\end{center}
%\end{figure}

\subsection{A stronger converse} \label{subsec:MTSCconverseRemarks}
For the CEO problem, we are able to obtain a stronger converse result as discussed in Section \ref{subsec:ceoStrongerConverse}.  We can obtain a similar result for the multiterminal source coding problem.  Indeed, the converse result we just proved continues to hold even when $\hat{Y}_i^n$ is allowed to be any probability measure on $\mathcal{Y}_i^n$, rather than a product distribution.  The proof of this fact is somewhat involved and can be found in Appendix \ref{app:strengthenedConverse}.

%\subsubsection{A Brief Remark on the Proof of the Strengthened Converse}
We note that the proof of this strengthened converse result (i.e., Theorem \ref{thm:strongConverse} in Appendix \ref{app:strengthenedConverse}) offers a direct proof of the converse of Theorem \ref{thm:MTSCregionDesc}, and as such we do not require a CEO result (Theorem \ref{thm:CEOregion}) or a ``black box" tuning argument (Lemma \ref{lem:amplifyConvex}).  At the heart of this alternative proof lies the Csisz\'{a}r sum identity (and a careful choice of auxiliary random variables) which provides a coupling between the attainable distortions for each source.  In the original proof of Theorem \ref{thm:MTSCregionDesc}, this coupling is accomplished by the tuning argument through Lemma \ref{lem:amplifyConvex}.

Interestingly, the two proofs are similar in spirit, with the key differences being the use of the Csisz\'{a}r sum identity versus the tuning argument.  Intuitively, the original tuning argument allows a ``clumsier" choice of auxiliary random variables which leads to a more elegant and transparent proof, but appears incapable of establishing the strengthened converse.  On the other hand, applying the Csisz\'{a}r sum identity requires a very careful choice of auxiliary random variables which, in turn,  affords a finer degree of control over various quantities. %Ultimately, this allows us to prove the strengthened converse given by Theorem \ref{thm:strongConverse}.

\subsection{An Example: The Daily Double} \label{subsec:exHorse}
The \emph{Daily Double} is a single bet that links together wagers on the winners of two consecutive horse races.  Winning the Daily Double is dependent on both wagers winning together.  In general, the outcomes of two consecutive races can be correlated (e.g. due to track conditions), so a gambler can potentially use this information to maximize his expected winnings.  Let $\mathcal{Y}_1$ and $\mathcal{Y}_2$ be the set of horses running in the first and second races respectively.  If horses $y_1$ and $y_2$ win their respective races, then the payoff is  $o(y_1,y_2)$ dollars for each dollar invested in outcome $(Y_1,Y_2)=(y_1,y_2)$.

There are two betting strategies one can follow:
\begin{enumerate}
\item The gambler can wager a fraction $b_1(y_1)$ of his wealth on horse $y_1$ winning the first race and parlay his winnings by betting a fraction $b_2(y_2)$ of his wealth on horse $y_2$ winning the second race.  In this case, the gambler's wealth relative is $b_1(Y_1)b_2(Y_2)o(Y_1,Y_2)$ upon learning the outcome of the Daily Double.  We refer to this betting strategy as the \emph{product-wager}.
\item The gambler can wager a fraction $b(y_1,y_2)$ of his wealth on horses $(y_1,y_2)$ winning the first and second races, respectively. In this case, the gambler's wealth relative is $b(Y_1,Y_2)o(Y_1,Y_2)$ upon learning the outcome of the Daily Double.  We refer to this betting strategy as the \emph{joint-wager}.
\end{enumerate}
Clearly the joint-wager  includes the product-wager as a special case.  However, the product-wager requires less effort to place, so the question is: how do the two betting strategies compare? 

To make things interesting, suppose the gamblers have access to rate-limited information about the first and second race outcomes at rates $R_1,R_2$ respectively.  Further, assume that $R_1\leq H(Y_1)$, $R_2\leq H(Y_2)$, and $R_1+R_2\leq H(Y_1,Y_2)$. For $(R_1,R_2)$ and $p(y_1,y_2)$ given, let $\mathcal{P}(R_1,R_2)$ denote the set of joint pmf's of the form
\begin{align*}
p(q,y_1,y_2,u_1,u_2)=p(q)p(y_1,y_2)p(u_1|y_1,q)p(u_1|y_1,q)
\end{align*}
which satisfy
\begin{align*}
R_1 &\geq I(Y_1;U_1|U_2,Q) \\
R_2 &\geq I(Y_2;U_2|U_1,Q)\\
R_1+R_2 &\geq I(Y_1,Y_2;U_1,U_2|Q)
\end{align*}
for  alphabets $\mathcal{U}_1,\mathcal{U}_2,\mathcal{Q}$ satisfying $|\mathcal{U}_i|\leq |\mathcal{Y}_i|$ and $|Q|\leq 5$. 

Typically, the quality of a bet is measured by the associated doubling rate (cf. \cite{bib:CoverThomas2006}). Theorem \ref{thm:MTSCregionDesc} implies that the optimal doubling rate for the product-wager is given by:
\begin{align*}
W_{\mbox{p-w}}^*(p(y_1,y_2))&=\sum_{y_1,y_2} p(y_1,y_2) \log b^*_1(y_1) b^*_2(y_2) o(y_1,y_2)\\
&=\mathbb{E} \log o(Y_1,Y_2) - \inf_{p\in\mathcal{P}(R_1,R_2)}\left\{H(Y_1|U_1,U_2,Q) + H(Y_2|U_1,U_2,Q)\right\}.
\end{align*}
Likewise, Theorem \ref{thm:SWminusD} implies that the optimal doubling rate for the joint-wager is given by:
\begin{align*}
W_{\mbox{j-w}}^*(p(y_1,y_2))&=\sum_{y_1,y_2} p(y_1,y_2) \log b^*(y_1,y_2) o(y_1,y_2)\\
&=\mathbb{E} \log o(Y_1,Y_2) +\min\{R_1-H(Y_1|Y_2),R_2-H(Y_2|Y_1), \\
&\quad\quad\quad\quad\quad\quad\quad\quad\quad\quad\quad\quad   R_1+R_2-H(Y_1,Y_2) \}.
\end{align*}

It is important to note that we do not require the side informations to be the same for each type of wager, rather, the side informations are only provided at the same rates.  Thus, the gambler placing the joint-wager receives side information at rates $(R_1,R_2)$ that maximizes his doubling rate, while the gambler placing the product-wager receives (potentially different) side information at rates $(R_1,R_2)$ that maximizes his doubling rate.  However, as we will see shortly, for any rates $(R_1,R_2)$, there always exists rate-limited side information which simultaneously allows each type of gambler to attain their maximum doubling rate.

By combining the expressions for $W_{\mbox{p-w}}^* (  p(y_1,y_2))$ and $W_{\mbox{j-w}}^*(p(y_1,y_2))$, we find that the difference in doubling rates is given by:
\begin{align}
&\Delta(R_1,R_2)
= W_{\mbox{j-w}}^* (  p(y_1,y_2)) -W_{\mbox{p-w}}^*(p(y_1,y_2))\notag\\
&= \min\Big\{R_1-H(Y_1|Y_2),R_2-H(Y_2|Y_1),R_1+R_2-H(Y_1,Y_2) \Big\}\notag\\
&\quad + \inf_{p\in\mathcal{P}(R_1,R_2)}\left\{H(Y_1|U_1,U_2,Q) + H(Y_2|U_1,U_2,Q)\right\}\label{eqn:Inf}\\
&= \hspace{-10pt} \inf_{p \in \mathcal{P}(R_1,R_2)} \hspace{-5pt} \min\Big\{ R_1-I(Y_1;U_1|U_2,Q)+I(Y_1;Y_2) -I(Y_1;U_2,Q) + H(Y_2|U_1,U_2,Q),\notag\\
&\quad R_2-I(Y_2;U_2|U_1,Q)+I(Y_2;Y_1)-I(Y_2;U_1,Q) + H(Y_1|U_1,U_2,Q),\notag\\
&\quad R_1+R_2-I(Y_1,Y_2;U_1,U_2|Q) + I(Y_1;Y_2|U_1,U_2,Q)\Big\}\notag\\
&= \hspace{-10pt} \inf_{p \in \mathcal{P}(R_1,R_2)} I(Y_1;Y_2|U_1,U_2,Q).\label{eqn:complicatedEqual}
\end{align}
The final equality \eqref{eqn:complicatedEqual} follows since 
\begin{itemize}
\item $R_1 \geq I(Y_1;U_1|U_2,Q)$ and $R_2 \geq I(Y_2;U_2|U_1,Q)$ for any $p\in\mathcal{P}(R_1,R_2)$.
\item $I(Y_2;Y_1)\geq I(Y_2;U_1,Q)$ and $I(Y_1;Y_2)\geq I(Y_1;U_2,Q)$ for any $p\in\mathcal{P}(R_1,R_2)$ by the data processing inequality.
\item The infimum in \eqref{eqn:Inf} is attained by a $p\in \mathcal{P}(R_1,R_2)$ satisfying $R_1+R_2=I(Y_1,Y_2;U_1,U_2|Q)$.  See Lemma \ref{lem:dailyDouble} in Appendix \ref{app:dailyDoubleLemma} for details.
\item By definition of conditional mutual information, 
\begin{align*}
H(Y_i|U_1,U_2,Q)\geq I(Y_1;Y_2|U_1,U_2,Q)
\end{align*}
for $i=1,2$.
\end{itemize}

Let $p^*\in \mathcal{P}(R_1,R_2)$ be the distribution that attains the infimum in \eqref{eqn:Inf} (such a $p^*$ always exists),  then \eqref{eqn:complicatedEqual} yields
\begin{align*}
& W_{\mbox{j-w}}^*  (  p(y_1,y_2)) -W_{\mbox{p-w}}^*(p(y_1,y_2)) \\
&= \sum_{u_1,u_2,q}p^*(u_1,u_2,q) \sum_{y_1,y_2} p^*(y_1,y_2|u_1,u_2,q) 
\log \frac{p^*(y_1,y_2|u_1,u_2,q)}{p^*(y_1|u_1,u_2,q)p^*(y_2|u_1,u_2,q)}\\
&= \mathbb{E}_{p^*} \log o(Y_1,Y_2) p^*(Y_1,Y_2|U_1,U_2,Q)\\
&\quad -\mathbb{E}_{p^*} \log o(Y_1,Y_2) p^*(Y_1|U_1,U_2,Q)p^*(Y_2|U_1,U_2,Q).
\end{align*}
Hence, we can interpret the auxiliary random variables corresponding to $p^*$ as optimal rate-limited side informations for \emph{both} betting strategies.  Moreover, optimal bets for each strategy are given by 
\begin{enumerate}
\item $b^*(y_1,y_2) = p^*(y_1,y_2|u_1,u_2,q)$ for the joint-wager, and
\item $b_1^*(y_1)=p^*(y_1|u_1,u_2,q), ~b_2^*(y_2)=p^*(y_2|u_1,u_2,q)$ for the product-wager.
\end{enumerate}

Since $\mathcal{P}(R_1,R_2) \subseteq \mathcal{P}(R_1',R_2')$ for $R_1\leq R_1'$ and $R_2\leq R_2'$, the function $\Delta(R_1,R_2)$ is nonincreasing in $R_1$ and $R_2$.  Thus, the benefits of using the joint-wager over the product-wager diminish in the amount of side-information available.  It is also not difficult to show that $\Delta(R_1,R_2)$ is jointly convex in $(R_1,R_2)$.

Furthermore, for rate-pairs $(R_1,R_2)$ and $(R_1',R_2')$ satisfying $R_1 < R_1'$ and $R_2 < R_2'$, there exist corresponding optimal joint- and product-wagers $b^{*}(y_1,y_2)$ and $b_1^{*}(y_1)b_2^{*}(y_2)$, and $b^{*'}(y_1,y_2)$ and $b_1^{*'}(y_1)b_2^{*'}(y_2)$, respectively, satisfying 
\begin{align}
D\Big(b^{*'}(y_1,y_2)\Big| \Big| b_1^{*'}(y_1)b_2^{*'}(y_2)\Big) < D\Big(b^{*}(y_1,y_2)\Big| \Big| b_1^{*}(y_1)b_2^{*}(y_2)\Big).\label{eqn:strictIneqD}
\end{align}
So, roughly speaking, the joint-wager and product-wager look ``more alike" as the amount of side information is increased.  
%Thus, in essence, any side information serves to decouple the dependence between the outcomes of the two races.  
The proof of the strict inequality in \eqref{eqn:strictIneqD} can be inferred from the proof of Lemma \ref{lem:dailyDouble} in Appendix \ref{app:dailyDoubleLemma}.

%If there exists a $U_2$ such that $Y_1\rightarrow Y_2 \rightarrow U_2$ and 
%\begin{align*}
%R_1 &\geq H(Y_1|U_2)\\
%R_2 &\geq I(Y_2;U_2),
%\end{align*}
%then $\Delta(R_1,R_2)=0$.  This can be observed by taking $U_1=Y_1$ and verifying that the corresponding distribution is in $\mathcal{P}(R_1,R_2)$. The interpretation here is that $(R_1,R_2)$ allow lossless reconstruction of $Y_1$ (by the source coding with side information result of Ahlswede-Wyner-K\"{o}rner), rendering the races conditionally independent given the side information.  Since symmetric statement holds also, if there exist auxiliary random variables $U_1$ and $U_2$ such that $Y_2\rightarrow Y_1 \rightarrow U_1$ and $Y_1\rightarrow Y_2 \rightarrow U_2$ and
%\begin{align*}
%R_1 &\geq \theta H(Y_1|U_2)+ (1-\theta)I(Y_1;U_1) \\
%R_2 &\geq \theta I(Y_2;U_2)+ (1-\theta)H(Y_2|U_1),
%\end{align*}
%for some $\theta\in[0,1]$ then $\Delta(R_1,R_2)=0$ by the convexity of $\Delta(R_1,R_2)$.

To conclude this example, we note that $\Delta(R_1,R_2)$ enjoys a great deal of symmetry near the origin in the sense that side information from either encoder contributes approximately the same amount to the improvement of the product-wager. We state this formally as a theorem:

\begin{theorem}
Define $\rho_m(Y_1,Y_2)$ to be the Hirschfeld-Gebelein-R\'{e}nyi maximal correlation between random variables $Y_1$ and $Y_2$.  Then,  $\Delta(R_1,R_2) \geq I(Y_1;Y_2)-\rho^2_m(Y_1,Y_2)\cdot(R_1+R_2)$. Moreover, this bound is tight as $(R_1,R_2)\rightarrow (0,0)$.
\end{theorem}
\begin{proof}
If $R_2=0$, then it is readily verified that $\Delta(R_1,0)$ can be expressed as follows:
\begin{align*}
\Delta(R_1,0) & =I(Y_1;Y_2)-\max_{\substack{p(u_1|y_1): I(U_1;Y_1)=R_1,\\
 U_1\rightarrow Y_1\rightarrow Y_2,~ |\mathcal{U}_1|\leq |\mathcal{Y}_1|+1}} I(U_1;Y_2).
\end{align*}
By symmetry:
\begin{align*}
\Delta(0,R_2) & =I(Y_1;Y_2)-\max_{\substack{p(u_2|y_2): I(U_2;Y_2)=R_2,\\
 U_2\rightarrow Y_2\rightarrow Y_1,~ |\mathcal{U}_2|\leq |\mathcal{Y}_2|+1}} I(U_2;Y_1).
\end{align*}
Here, we can apply a result of Erkip \cite[Theorem 10]{bib:Erkip1996} to evaluate the gradient of $\Delta(R_1,R_2)$ at $(R_1,R_2)=(0,0)$:
\begin{align}
\left. \frac{\partial}{\partial R_1} \Delta(R_1,R_2) \right|_{(R_1,R_2)=(0,0)} = \left. \frac{\partial}{\partial R_2} \Delta(R_1,R_2) \right|_{(R_1,R_2)=(0,0)} = -\rho^2_m(Y_1,Y_2). \label{eqn:grad}
\end{align}
Note, since $\Delta(R_1,0)$ and $\Delta(0,R_2)$ are each convex in their respective variable and $\Delta(0,0)=I(Y_1;Y_2)$, we have
\begin{align}
\Delta(R_1,0) &\geq I(Y_1;Y_2) -\rho^2_m(Y_1,Y_2)R_1 \notag\\
\Delta(0,R_2) &\geq I(Y_1;Y_2) -\rho^2_m(Y_1,Y_2)R_2 \label{eqn:sepCvx2}.
\end{align}

Taking this one step further, for $\nu_1,\nu_2>0$, we can evaluate the one-sided derivative:
\begin{align}
\lim_{\lambda \downarrow 0 } \frac{\Delta(\lambda \nu_1,\lambda \nu_2)-\Delta(0,0)}{\lambda} = -\rho_m^2(Y_1,Y_2) \cdot (\nu_1+\nu_2).
\label{eqn:directDeriv}
\end{align}
We remark that \eqref{eqn:directDeriv} does not follow immediately from \eqref{eqn:grad} since the point at which we are taking the derivatives (i.e., the origin) does not lie in an open neighborhood of the domain.  Nonetheless, the expected result holds.

Since $\Delta(R_1,R_2)$ is convex, we obtain an  upper bound on the one-sided derivative as follows:
%given by the supporting hyperplane tangent to $\Delta(R_1,R_2)$ at $(R_1,R_2)=(0,0)$:
%With this observation in mind, we note that
\begin{align*}
\lim_{\lambda \downarrow 0 } \frac{\Delta(\lambda \nu_1,\lambda \nu_2)-\Delta(0,0)}{\lambda} &\leq \lim_{\lambda \downarrow 0 } \frac{\frac{1}{2} \Delta(2\lambda\nu_1,0)+\frac{1}{2} \Delta(0,2\lambda \nu_2)- \Delta(0,0;p)}{\lambda} \\
&= \frac{1}{2}\lim_{\lambda \downarrow 0 } \frac{\Delta(\lambda 2\nu_1,0)- \Delta(0,0)}{\lambda} \\
&\quad+ \frac{1}{2}\lim_{\lambda \downarrow 0 } \frac{\Delta(0,\lambda 2 \nu_2)- \Delta(0,0)}{\lambda}\\
%&= -\frac{1}{2}\rho_m(Y_1,Y_2)^2 \cdot (2\nu_1+2\nu_2)\\
&=-\rho_m^2(Y_1,Y_2) \cdot (\nu_1+\nu_2),
\end{align*}
where the final equality follows by \eqref{eqn:grad} and the positive homogeneity of the directional derivative.

Therefore, to complete the proof of \eqref{eqn:directDeriv}, it suffices to prove the lower bound 
\begin{align*}
\lim_{\lambda \downarrow 0 } \frac{\Delta(\lambda \nu_1,\lambda \nu_2)-\Delta(0,0)}{\lambda} \geq -\rho_m^2(Y_1,Y_2) \cdot (\nu_1+\nu_2).
\end{align*}
To this end, fix $\lambda,\nu_1,\nu_2>0$ and observe that
%
%\begin{align*}
%\Delta(R_1,R_2) \geq I(Y_1;Y_2)-\rho^2_m(Y_1,Y_2)\cdot(R_1+R_2).
%\end{align*}
%By \eqref{eqn:grad} and convexity of $\Delta(R_1,R_2)$, this bound is tight as $(R_1,R_2)\rightarrow (0,0)$.
%
%To see that \eqref{eqn:directDeriv} holds, it suffices to prove the lower bound 
%\begin{align*}
%\lim_{\lambda \downarrow 0 } \frac{\Delta(\lambda \nu_1,\lambda \nu_2)-\Delta(0,0)}{\lambda} \geq -\rho_m(Y_1,Y_2)^2 \cdot (\nu_1+\nu_2)
%\end{align*}
%since the reverse inequality follows from convexity of $\Delta(R_1,R_2)$ and the known derivatives \eqref{eqn:grad} along the coordinate axis.
%To this end, fix $\lambda,\nu_1,\nu_2>0$ and observe that
\begin{align}
& \frac{\Delta(\lambda \nu_1,\lambda \nu_2)-\Delta(0,0)}{\lambda} \notag \\
&\quad=
\frac{1}{\lambda}\inf_{p\in\mathcal{P}(\lambda \nu_1,\lambda \nu_2)}\Big\{ 
I(Y_1;Y_2|U_1,U_2|Q)-I(Y_1;Y_2)
\Big\}\label{eqn:defnDelta}\\
&\quad=
\frac{1}{\lambda}\inf_{p\in\mathcal{P}(\lambda \nu_1,\lambda \nu_2)}\Big\{ I(Y_1,Y_2;U_1,U_2|Q)-I(Y_1;U_1,U_2|Q)-I(Y_2;U_1,U_2|Q)\Big\}\notag\\
&\quad=(\nu_1+\nu_2)-\frac{1}{\lambda}\Big(I_{p^*}(Y_1;U_1,U_2|Q)+I_{p^*}(Y_2;U_1,U_2|Q)\Big)\label{eqn:ddLemmaCor} \\
&\quad= (\nu_1+\nu_2)-\frac{1}{\lambda}\Big(I_{p^*}(Y_1;U_1|U_2,Q)+I_{p^*}(Y_1;U_2|Q)\notag \\
&\quad\quad\quad\quad\quad\quad\quad\quad\quad +I_{p^*}(Y_2;U_2|U_1,Q)+I_{p^*}(Y_2;U_1|Q)\Big)\notag \\
&\quad \geq (\nu_1+\nu_2)-\rho_m^2(Y_1,Y_2)\left(  2 \nu_1 + 2 \nu_2  \right)\notag\\
&\quad\quad\quad -\frac{(1-\rho_m^2(Y_1,Y_2))}{\lambda}  \left(I_{p^*}(Y_1;U_1|U_2,Q) +I_{p^*}(Y_2;U_2|U_1,Q) \right) \label{eqn:trickyStep}\\
&\quad = -\rho_m^2(Y_1,Y_2)\left(  \nu_1 + \nu_2  \right) + (1-\rho_m^2(Y_1,Y_2))\left(  \nu_1 + \nu_2  \right)\notag\\
&\quad\quad\quad -\frac{(1-\rho_m^2(Y_1,Y_2))}{\lambda}  \left(I_{p^*}(Y_1;U_1|U_2,Q) +I_{p^*}(Y_2;U_2|U_1,Q) \right)\notag\\
&\quad\geq-\rho_m^2(Y_1,Y_2)\left(  \nu_1 + \nu_2  \right). \label{eqn:lastStep}
%
%
%
%&\quad\geq (\nu_1+\nu_2)-\frac{1}{\lambda}\Big(\lambda \nu_1 +I_{p^*}(Y_1;U_2|Q)+\lambda \nu_2 + I_{p^*}(Y_2;U_1|Q)\Big)\label{eqn:defOfP}\\
%&\quad=-\frac{1}{\lambda}\Big(I_{p^*}(Y_1;U_2|Q)+I_{p^*}(Y_2;U_1|Q)\Big)\label{eqn:maxSepPrev} \\
%&\quad\geq -\frac{1}{\lambda}\left( \max_{\substack{p(\tilde{u}_2|y_2): I(\tilde{U}_2;Y_2)\leq \lambda \nu_2,\\
% \tilde{U}_2\rightarrow Y_2\rightarrow Y_1,~ |\tilde{\mathcal{U}}_2|\leq |\mathcal{Y}_2|+1}} \hspace{-10pt} I(\tilde{U}_2;Y_1)
% ~ +  \hspace{-10pt} \max_{\substack{p(\tilde{u}_1|y_1): I(\tilde{U}_1;Y_1)\leq \lambda \nu_1,\\
% \tilde{U}_1\rightarrow Y_1\rightarrow Y_2,~ |\tilde{\mathcal{U}}_1|\leq |\mathcal{Y}_1|+1}}  \hspace{-10pt} I(\tilde{U}_1;Y_2)  \right)\label{eqn:maxSep}\\
% &\quad\underset{\lambda \downarrow 0}{\longrightarrow} \left. \nu_2 \frac{\partial}{\partial R_2} \Delta(R_1,R_2) \right|_{(R_1,R_2)=(0,0)} + \nu_1 \left. \frac{\partial}{\partial R_1} \Delta(R_1,R_2) \right|_{(R_1,R_2)=(0,0)}\label{eqn:limitLambda}\\
% &\quad= -\rho_m^2(Y_1,Y_2) \cdot (\nu_1+\nu_2).\label{eqn:grad2}
\end{align}
In the above string of inequalities
\begin{itemize}
\item \eqref{eqn:defnDelta} follows by definition of $\Delta(R_1,R_2)$.
\item Equality \eqref{eqn:ddLemmaCor} follows since Lemma \ref{lem:dailyDouble} guarantees that the infimum is attained in \eqref{eqn:defnDelta} for some $p^*\in \mathcal{P}(\lambda \nu_1,\lambda \nu_2)$ satisfying $I_{p^*}(Y_1,Y_2;U_1,U_2|Q)=\lambda(\nu_1+\nu_2)$.  Here, we write $I_{p^*}(Y_1,Y_2;U_1,U_2|Q)$ to denote the mutual information $I(Y_1,Y_2;U_1,U_2|Q)$ evaluated for the distribution $p^*$.
\item To see that \eqref{eqn:trickyStep} holds, note that
\begin{align*}
I_{p^*}(Y_2;U_2|Q) = \lambda \nu_1 + \lambda \nu_2 -I_{p^*}(Y_1;U_1|U_2,Q),
\end{align*}
%since $p^*\in \mathcal{P}(\lambda \nu_1,\lambda \nu_2)$, 
and thus
\begin{align}
&I(Y_1;Y_2)-\rho_m^2(Y_1,Y_2)\left(  \lambda \nu_1 + \lambda \nu_2 -I_{p^*}(Y_1;U_1|U_2,Q) \right)\notag\\
&\quad\leq \Delta(0,\lambda \nu_1 + \lambda \nu_2 -I_{p^*}(Y_1;U_1|U_2,Q)) \label{eqn:Step1}\\
&\quad=I(Y_1;Y_2)-\max_{\substack{p(\tilde{u}_2|y_2): I(Y_2;\tilde{U}_2)\leq \lambda \nu_1 + \lambda \nu_2 -I_{p^*}(Y_1;U_1|U_2,Q), \\
 \tilde{U}_2\leftrightarrow Y_2\leftrightarrow Y_1}} I(\tilde{U}_2;Y_1)  \label{eqn:Step2} \\
&\quad\leq I(Y_1;Y_2)-I_{p^*}(Y_1;U_2|Q), \label{eqn:Step3}
\end{align}
which implies
\begin{align*}
-\rho_m^2(Y_1,Y_2) \left(  \lambda \nu_1 + \lambda \nu_2 -I_{p^*}(Y_1;U_1|U_2,Q) \right) \leq -I_{p^*}(Y_1;U_2|Q).
\end{align*}
The above steps are justified as follows:
\begin{itemize}
\item \eqref{eqn:Step1} follows from \eqref{eqn:sepCvx2}.
\item \eqref{eqn:Step2} follows by definition of the function $\Delta(0,x)$.
\item \eqref{eqn:Step3} follows since $Q$ is independent of $Y_1,Y_2$ (by definition of $p^*$), and thus $\tilde{U}_2=(U_2,Q)$ lies in the set over which we take the maximum in \eqref{eqn:Step2}.
\end{itemize}
By symmetry, we conclude that
\begin{align*}
&-(I_{p^*}(Y_1;U_2|Q)+I_{p^*}(Y_2;U_1|Q))\\
&\quad\geq -\rho_m^2(Y_1,Y_2)\left(  2\lambda \nu_1 + 2\lambda \nu_2 -I_{p^*}(Y_1;U_1|U_2,Q)-I_{p^*}(Y_2;U_2|U_1,Q) \right),
\end{align*}
and \eqref{eqn:trickyStep} follows.
%
%\item \eqref{eqn:defOfP} follows since $I(Y_1;U_1|U_2,Q)\leq \lambda \nu_1$ and  $I(Y_2;U_2|U_1,Q)\leq \lambda \nu_2$ for any $p \in \mathcal{P}(\lambda \nu_1,\lambda \nu_2)$.
%\item \eqref{eqn:maxSep} follows since maximizing each term separately in \eqref{eqn:maxSepPrev} over a set of distributions that includes $\mathcal{P}(\lambda \nu_1,\lambda \nu_2)$ yields a lower bound.
%\item The limit \eqref{eqn:limitLambda} follows from \eqref{eqn:grad} and positive homogeneity of the one-sided derivatives.
%\item Finally, \eqref{eqn:grad2} follows from the known derivatives along the coordinate axis, given by \eqref{eqn:grad}.
\item   \eqref{eqn:lastStep} follows since $\lambda \nu_1 \geq I_{p^*}(Y_1;U_1|U_2,Q)$ and $\lambda \nu_2 \geq I_{p^*}(Y_2;U_2|U_1,Q)$ for $p^*\in \mathcal{P}(\lambda \nu_1,\lambda \nu_2)$.
\end{itemize}
\end{proof}

\subsection{An Application: List Decoding} \label{subsec:ListDecoding}
In the previous example, we did not take advantage of the stronger converse result which we proved in Appendix \ref{app:strengthenedConverse} (see the discussion in Section \ref{subsec:MTSCconverseRemarks}).  In this section, we give an application that requires this strengthened result.

Formally, a 2-list code (of blocklength $n$ consists) of encoding functions:
\begin{align*}
g_i^{(n)}: \mathcal{Y}_i^n \rightarrow \left\{ 1,\dots,M_i^{(n)}\right\} \mbox{~for $i=1,2$}
\end{align*}
and list decoding functions
\begin{align*}
&L_1^{(n)}: \left\{ 1,\dots,M_1^{(n)} \right\} \times \left\{ 1,\dots,M_2^{(n)} \right\} \rightarrow 2^{\mathcal{Y}_1^n}\\
&L_2^{(n)}: \left\{ 1,\dots,M_1^{(n)} \right\} \times \left\{ 1,\dots,M_2^{(n)} \right\} \rightarrow 2^{\mathcal{Y}_2^n}.
\end{align*}
A list decoding tuple $(R_1,R_2,\Delta_1,\Delta_2)$ is achievable if, for any $\epsilon>0$, there exists a 2-list code of blocklength $n$ satisfying the rate constraints
\begin{align*}
\frac{1}{n}\log M_1^{(n)} &\leq R_1+\epsilon\\
\frac{1}{n}\log M_2^{(n)} &\leq R_2+\epsilon,\\
\end{align*}
and the probability of list-decoding error constraints
\begin{align*}
\Pr\left[Y_1^n \notin L_1^{(n)}\left(g_1^{(n)}(Y_1^n),g_2^{(n)}(Y_2^n)\right) \right]&\leq \epsilon,\\
\Pr\left[Y_2^n \notin L_2^{(n)}\left(g_1^{(n)}(Y_1^n),g_2^{(n)}(Y_2^n)\right) \right]&\leq \epsilon.
\end{align*}
with list sizes
\begin{align*}
\frac{1}{n}\log|L_1^{(n)}|&\leq \Delta_1+\epsilon\\
\frac{1}{n}\log|L_2^{(n)}|&\leq \Delta_2+\epsilon.
\end{align*}
With a 2-list code so defined, the following theorem  shows that the 2-list decoding problem and  multiterminal source coding problem under logarithmic loss are equivalent (inasmuch as the achievable regions are identical):
\begin{theorem} \label{thm:listDecoding}
The list decoding tuple $(R_1,R_2,\Delta_1,\Delta_2)$ is achievable if and only if
\begin{align*}
R_1 &\geq I(U_1;Y_1|U_2,Q)\\
R_2 &\geq I(U_2;Y_2|U_1,Q)\\
R_1+R_2 &\geq I(U_1,U_2;Y_1,Y_2|Q)\\
\Delta_1 &\geq H(Y_1|U_1,U_2,Q)\\
\Delta_2 &\geq H(Y_2|U_1,U_2,Q).
\end{align*}
for some joint distribution 
\begin{align*}
p(y_1,y_2,u_1,u_2,q)=p(y_1,y_2)p(u_1|y_1,q)p(u_2|y_2,q)p(q),
\end{align*}
where $|\mathcal{U}_1|\leq |\mathcal{Y}_1|$, $|\mathcal{U}_2|\leq |\mathcal{Y}_2|$, and $|\mathcal{Q}|\leq 5$.
\end{theorem}
\begin{remark}
We note that a similar connection to list decoding can be made for other multiterminal scenarios, in particular the CEO problem.
\end{remark}
To prove the theorem, we require a slightly modified version of \cite[Lemma 1]{bib:YHKimSutivongCover2008}:
\begin{lemma}\label{lem:listDecoding}
If the list decoding tuple $(R_1,R_2,\Delta_1,\Delta_2)$ is achieved by a sequence of 2-list codes $\{g_1^{(n)},g_2^{(n)},L_1^{(n)},L_2^{(n)}\}_{n\rightarrow \infty}$, then 
\begin{align*}
H(Y_1^n|g_1^{(n)}(Y_1^n),g_2^{(n)}(Y_2^n)) &\leq |L_1^{(n)}| + n \epsilon_n\\
H(Y_2^n|g_1^{(n)}(Y_1^n),g_2^{(n)}(Y_2^n)) &\leq |L_2^{(n)}| + n \epsilon_n,
\end{align*}
where $\epsilon_n\rightarrow 0$ as $n\rightarrow \infty$.
\end{lemma}
\begin{proof}
The proof is virtually identical to that of \cite[Lemma 1]{bib:YHKimSutivongCover2008}, and is therefore omitted.
\end{proof}
\begin{proof}[Proof of Theorem \ref{thm:listDecoding}]
First observe that the direct part is trivial.  Indeed, for a joint distribution $p(y_1,y_2,u_1,u_2,q)=p(y_1,y_2)p(u_1|y_1,q)p(u_2|y_2,q)p(q)$, apply the Berger-Tung achievability scheme and take $L_i^{(n)}$ to be the set of $y_i^n$ sequences which are jointly typical with the decoded quantizations $(U_1^n,U_2^n)$.  This set has cardinality no larger than $2^{n(H(Y_i|U_1,U_2,Q)+\epsilon)}$, which proves achievability.  

To see the converse, note that setting
\begin{align*}
\hat{Y}_i^n = \Pr\left[Y_i^n | g_1^{(n)}(Y_1^n),g_2^{(n)}(Y_2^n)\right]
\end{align*}
achieves a logarithmic loss of $\frac{1}{n}H(Y_i^n|g_1^{(n)}(Y_1^n),g_2^{(n)}(Y_2^n))$ for source $i$ in the setting where reproductions are not restricted to product distributions.  Applying the strengthened converse of Theorem \ref{thm:MTSCregionDesc} together with Lemma \ref{lem:listDecoding} yields the desired result.
\end{proof}

\section{Relationship to the General Multiterminal Source Coding Problem} \label{sec:generalMTSC}

In this section, we relate our results for logarithmic loss to multiterminal source coding problems with arbitrary distortion measures and reproduction alphabets.

As before, we let $\left\{Y_{1,j},Y_{2,j}\right\}_{j=1}^n$ be a sequence of $n$ independent, identically distributed random variables with finite alphabets $\mathcal{Y}_1$ and $\mathcal{Y}_2$, respectively, and joint pmf $p(y_1,y_2)$.

In this section, the reproduction alphabets $\breve{\mathcal{Y}}_i$, $i=1,2$,  are arbitrary.  We also consider generic distortion measures:
\begin{align*}
\breve{d}_i:\mathcal{Y}_i \times \breve{\mathcal{Y}}_i \rightarrow \mathbb{R}^+ \mbox{~for $i=1,2$},
\end{align*}
where $\mathbb{R}^+$ denotes the set of nonnegative real numbers.  The sequence distortion is then defined as follows:
\begin{align*}
\breve{d}_i({y}_i^n,\breve{y}_i^n)=\frac{1}{n}\sum_{j=1}^n \breve{d}_i({y}_{i,j},\breve{y}_{i,j}).
\end{align*}

We will continue to let $d(\cdot,\cdot)$ and $\hat{\mathcal{Y}}_1,\hat{\mathcal{Y}}_2$ denote the logarithmic loss distortion measure and the associated reproduction alphabets, respectively.

A rate distortion code (of blocklength $n$) consists of encoding functions:
\begin{align*}
\breve{g}_i^{(n)}: \mathcal{Y}_i^n \rightarrow \left\{ 1,\dots,M_i^{(n)}\right\} \mbox{~for $i=1,2$}
\end{align*}
and decoding functions
\begin{align*}
\breve{\psi}_i^{(n)}: \left\{ 1,\dots,M_1^{(n)} \right\} \times \left\{ 1,\dots,M_2^{(n)} \right\} \rightarrow {\breve{\mathcal{Y}}_i}^n \mbox{~for $i=1,2$}.
\end{align*}

A rate distortion vector $(R_1,R_2,D_1,D_2)$ is strict-sense achievable if there exists a blocklength $n$, encoding functions $\breve{g}_1^{(n)},\breve{g}_2^{(n)}$ and a decoder $(\breve{\psi}_1^{(n)},
\breve{\psi}_2^{(n)})$ such that
\begin{align}
R_i &\geq \frac{1}{n}\log M_i^{(n)} \mbox{~for $i=1,2$} \label{eqn:genRate} \\
D_i &\geq  \mathbb{E} \breve{d}_i({Y}_i^n,\breve{Y}_i^n) \mbox{~for $i=1,2$}. \label{eqn:genDist}
\end{align}
Where 
\begin{align*}
\breve{Y}_i^n=\breve{\psi}_i^{(n)}(\breve{g}_1^{(n)}(Y_1^n),\breve{g}_2^{(n)}(Y_2^n)) \mbox{~for $i=1,2$}.
\end{align*}
For these functions, we define the quantity
\begin{align}
\beta_i\left(\breve{g}_1^{(n)},\breve{g}_2^{(n)},\breve{\psi}_1^{(n)},
\breve{\psi}_2^{(n)}\right):=\frac{1}{n}\sum_{j=1}^n \mathbb{E}\log
\left(\sum_{y_i\in \mathcal{Y}_i} 2^{-\breve{d}_i(y_i,\breve{Y}_{i,j})} \right) \mbox{~for $i=1,2$}.
\end{align}
Now, let $\beta_i(R_1,R_2,D_1,D_2)$ be the infimum of the $\beta_i\left(\breve{g}_1^{(n)},\breve{g}_2^{(n)},\breve{\psi}_1^{(n)},
\breve{\psi}_2^{(n)}\right)$'s, where the infimum is taken over all codes that achieve the rate distortion vector $(R_1,R_2,D_1,D_2)$.

At this point it is instructive to pause and consider some examples.
\begin{example}[Binary Sources and Hamming Distortion]
For $i=1,2$, let $\breve{\mathcal{Y}}_i=\mathcal{Y}_i=\{0,1\}$ and let $\breve{d}_i$ be the $\alpha$-scaled Hamming distortion measure:
\begin{align*}
\breve{d}_i({y}_i,\breve{y}_i)=\left\{
\begin{array}{ll}
0 & \mbox{if $\breve{y}_i=y_i$,}\\
\alpha & \mbox{if $\breve{y}_i\neq y_i$.}
\end{array}
 \right.
\end{align*}
In this case,
\begin{align}
\sum_{y_i\in \mathcal{Y}_i} 2^{-\breve{d}_i(y_i,\breve{Y}_{i,j})}=2^0+2^{-\alpha},
\end{align}
so $\beta_i(R_1,R_2,D_1,D_2)=\log(1+2^{-\alpha})$ for any $(R_1,R_2,D_1,D_2)$. This notion that $\beta_i(R_1,R_2,D_1,D_2)$ is a constant extends to all distortion measures for which the columns of the $|\mathcal{Y}_i|\times |\breve{\mathcal{Y}}_i|$ distortion matrix are permutations of one another.

\end{example}

\begin{example}[Binary Sources and Erasure Distortion]
For $i=1,2$, let $\mathcal{Y}_i=\{0,1\}$, $\breve{\mathcal{Y}}_i=\{0,1,e\}$ and let $\breve{d}_i$ be the standard erasure distortion measure:
\begin{align*}
\breve{d}_i({y}_i,\breve{y}_i)=\left\{
\begin{array}{ll}
0 & \mbox{if $\breve{y}_i=y_i$}\\
1 & \mbox{if $\breve{y}_i=e$}\\
\infty & \mbox{if $\breve{y}_i\in\{0,1\}$ and $\breve{y}_i\neq y_i$}.
\end{array}
 \right.
\end{align*}
In this case,
\begin{align}
\sum_{y_i\in \mathcal{Y}_i} 2^{-\breve{d}_i(y_i,\breve{Y}_{i,j})}=
\left\{
\begin{array}{ll}
2^{-\infty}+2^0=1 & \mbox{if $\breve{Y}_{i,j} \in \{0,1\}$}\\
2^{-1}+2^{-1}=1 & \mbox{if $\breve{Y}_{i,j}=e$}.
\end{array}
 \right.
\end{align}
so $\beta_i(R_1,R_2,D_1,D_2)=0$ for any $(R_1,R_2,D_1,D_2)$. This result can easily be extended to erasure distortion on larger alphabets by setting the penalty to $\log |\mathcal{Y}_i|$ when $\breve{Y}_i=e$.
\end{example}
%
%Let $\mathcal{RD}^{\star}_{\mbox{gen}}$ denote the set of strict-sense achievable rate distortion vectors for this general multiterminal source coding problem and define the set of achievable rate distortion vectors to be its closure, $\overline{\mathcal{RD}}^{\star}_{\mbox{gen}}$.
%Our first main result of this section is an outer bound on $\overline{\mathcal{RD}}^{\star}_{\mbox{gen}}$.

\begin{theorem}\label{thm:outerBoundGeneral}
Suppose $(R_1,R_2,D_1,D_2)$ is strict-sense achievable for the general multiterminal source coding problem.  Then
\begin{align}
\left.
\begin{array}{rl}
R_1 &\geq I(U_1;Y_1|U_2,Q) \\
R_2 &\geq I(U_2;Y_2|U_1,Q)\\
R_1+R_2 &\geq I(U_1,U_2;Y_1,Y_2|Q)\\
D_1 &\geq H(Y_1|U_1,U_2,Q) - \beta_1(R_1,R_2,D_1,D_2) \\
D_2 &\geq H(Y_2|U_1,U_2,Q) - \beta_2(R_1,R_2,D_1,D_2) %\label{eqn:outerBoundEq5}
\end{array} \right\} \label{eqn:outerBoundEq1}
\end{align}
for some joint distribution $p(y_1,y_2)p(q)p(u_1|y_1,q)p(u_2|y_2,q)$ with $|\mathcal{U}_i|\leq |\mathcal{Y}_i|$ and $|\mathcal{Q}|\leq 5$.
\end{theorem}
\begin{proof}
Since $(R_1,R_2,D_1,D_2)$ is strict-sense achievable, there exists a blocklength $n$, encoding functions $\breve{g}_1^{(n)},\breve{g}_2^{(n)}$ and a decoder $(\breve{\psi}_1^{(n)},
\breve{\psi}_2^{(n)})$ satisfying \eqref{eqn:genRate}-\eqref{eqn:genDist}.  Given these functions, the decoder can generate reproductions $\breve{Y}_1^n,\breve{Y}_2^n$ satisfying the average distortion constraints \eqref{eqn:genDist}.  From the reproduction $\breve{Y}_i^n$, we construct the reproduction $\hat{Y}_i^n$ as follows:
\begin{align*}
\hat{Y}_j(y_i)=\frac{ 2^{-\breve{d}_i(y_i,\breve{Y}_{i,j})}}{\sum_{y_i'\in \mathcal{Y}_i} 2^{-\breve{d}_i(y_i',\breve{Y}_{i,j})}}.
\end{align*}
Now, using the logarithmic loss distortion measure, observe that $\hat{Y}_i^n$ satisfies
\begin{align*}
\mathbb{E} d({Y}_i^n,\hat{Y}_i^n) &= 
\frac{1}{n}\sum_{j=1}^n \mathbb{E} \log\left(2^{\breve{d}_i({Y}_{i,j},\breve{Y}_{i,j})} \right) + \frac{1}{n}\sum_{j=1}^n \mathbb{E} \log \left( \sum_{y_i'\in \mathcal{Y}_i} 2^{-\breve{d}_i(y_i',\breve{Y}_{i,j})} \right) \\
&= \frac{1}{n}\sum_{j=1}^n \mathbb{E}\breve{d}_i({Y}_{i,j},\breve{Y}_{i,j}) + \beta_i\left(\breve{g}_1^{(n)},\breve{g}_2^{(n)},\breve{\psi}_1^{(n)}, \breve{\psi}_2^{(n)}\right)\\
&\leq D_i +\beta_i\left(\breve{g}_1^{(n)},\breve{g}_2^{(n)},\breve{\psi}_1^{(n)}, \breve{\psi}_2^{(n)}\right)\\
&:=\tilde{D}_i.
\end{align*}
Thus, $(R_1,R_2,\tilde{D}_1,\tilde{D}_2)$ is achievable for the multiterminal source coding problem with the logarithmic loss distortion measure.  Applying Theorem \ref{thm:MTSCregionDesc} and taking the infimum over all coding schemes that achieve $(R_1,R_2,D_1,D_2)$ proves the theorem.
\end{proof}

This outer bound is interesting because the region is defined over the same set of probability distributions that define the Berger-Tung inner bound.  While the $\beta_i$'s can be difficult to compute in general, we have shown that they can be readily determined for many popular distortion measures.  As an application,  we now give a quantitative approximation of the rate distortion region for binary sources subject to Hamming distortion constraints.  Before proceeding, we prove the following lemma.

\begin{lemma} \label{lem:BTalpha}
Suppose $(R_1,R_2,\tilde{D}_1,\tilde{D}_2)$ is strict-sense achievable for the multiterminal source coding problem with binary sources and $\breve{d}_i$ equal to the $\alpha_i$-scaled Hamming distortion measure, for $i=1,2$.  Then the Berger-Tung achievability scheme can achieve a point $(R_1,R_2,D_1,D_2)$ satisfying
\begin{align*}
D_i-\tilde{D}_i \leq \left(\frac{\alpha_i}{2}-1\right) H_i +\log(1+2^{-\alpha_i})
\end{align*}
for some $H_i \in[0,1]$, $i=1,2$.
\end{lemma}
\begin{proof}
By Theorem \ref{thm:outerBoundGeneral}, $(R_1,R_2,\tilde{D}_1,\tilde{D}_2)$ satisfy \eqref{eqn:outerBoundEq1} for some joint distribution \linebreak $p(y_1,y_2)p(q)p(u_1|y_1,q)p(u_2|y_2,q)$.  For this distribution, define the reproduction functions 
\begin{align}
\breve{Y}_i(U_1,U_2,Q)=\arg \max_{y_i} p(y_i |U_1,U_2,Q) \mbox{~for $i=1,2$}.
\end{align}
Then, observe that for $i=1,2$:
\begin{align}
\mathbb{E} \breve{d}_i({Y}_i,\breve{Y}_i) &= 
\sum_{u_1,u_2,q} p(u_1,u_2,q) \left[ \alpha_i \cdot \min_{y_i} p(y_i|u_1,u_2,q) + 0\cdot \max_{y_i} p(y_i|u_1,u_2,q)\right]\notag \\
&= \alpha_i \sum_{u_1,u_2,q} p(u_1,u_2,q) \cdot \min_{y_i} p(y_i|u_1,u_2,q) \notag \\
&\leq \frac{\alpha_i}{2} \sum_{u_1,u_2,q} p(u_1,u_2,q) \cdot H(Y_i|U_1,U_2,Q=u_1,u_2,q) \label{eqn:entropyUB}\\
&=\frac{\alpha_i}{2} H(Y_i|U_1,U_2,Q). \notag %\label{eqn:distUBHam}
\end{align}
Where \eqref{eqn:entropyUB} follows from the fact that $2p\leq h_2(p)$ for $0\leq p\leq 0.5$.  Thus, $D_i=\frac{\alpha_i}{2} H(Y_i|U_1,U_2,Q)$ is achievable for rates $(R_1,R_2)$ using the Berger-Tung achievability scheme.  Combining this with the fact that $\tilde{D}_i \geq H(Y_i|U_1,U_2,Q) - \log (1+2^{-\alpha_i})$, we see that
\begin{align*}
D_i-\tilde{D}_i\leq \frac{\alpha_i}{2} H(Y_i|U_1,U_2,Q) - H(Y_i|U_1,U_2,Q) + \log (1+2^{-\alpha_i}). 
\end{align*}
\end{proof}

Lemma \ref{lem:BTalpha} allows us to give a quantitative outer bound on the achievable rate distortion region in terms of the Berger-Tung inner bound.

\begin{corollary} \label{cor:BTbinHam}
Suppose $(R_1,R_2,\tilde{D}_1^{(1)},\tilde{D}_2^{(1)})$ is strict-sense achievable for the multiterminal source coding problem with binary sources and $\breve{d_i}$ equal to the standard $1$-scaled Hamming distortion measure, for  $i=1,2$.  Then the Berger-Tung achievability scheme can achieve a point $(R_1,R_2,D_1^{(1)},D_2^{(1)})$, where
\begin{align*}
D_i^{(1)}-\tilde{D}_i^{(1)} \leq  \frac{1}{2}\log\left(\frac{5}{4}\right) < 0.161  \mbox{~for $i=1,2$}.
\end{align*}
\end{corollary}
\begin{proof}
For rates $(R_1,R_2)$, note that distortions $(\tilde{D}_1,\tilde{D}_2)$ are strict-sense achievable for the $\alpha_i$-scaled Hamming distortion measures if and only if distortions  $(\tilde{D}_1^{(1)},\tilde{D}_2^{(1)})=(\frac{1}{\alpha_1}\tilde{D}_1,\frac{1}{\alpha_2}\tilde{D}_2)$ are strict-sense achievable for the $1$-scaled Hamming distortion measure.  Likewise, the point $(R_1,R_2,D_1,D_2)$ is achieved by the Berger-Tung coding scheme for the $\alpha_i$-scaled Hamming distortion measures if and only if $(R_1,R_2,\frac{1}{\alpha_1}D_1,\frac{1}{\alpha_2}D_2)$ is achieved by the Berger-Tung coding scheme for the $1$-scaled Hamming distortion measure.  

Thus, applying Lemma \ref{lem:BTalpha}, we can use the Berger-Tung achievability scheme to achieve a point $(R_1,R_2,D_1^{(1)},D_2^{(1)})$ satisfying 
\begin{align}
D_i^{(1)}-\tilde{D}_i^{(1)} &= \frac{1}{\alpha_i}\left( D_i-\tilde{D}_i \right)\notag\\
&\leq \frac{1}{\alpha_i} \left(\frac{\alpha_i}{2}-1\right) H_i +\frac{1}{\alpha_i} \log(1+2^{-\alpha_i}) \notag\\
&= \left(\frac{1}{2}-\frac{1}{\alpha_i} \right) H_i +\frac{1}{\alpha_i} \log(1+2^{-\alpha_i}) \label{eqn:minExpr}
\end{align}
for some $H_i\in[0,1]$.  We can optimize \eqref{eqn:minExpr} over $\alpha_i$ to find the minimum gap for a given $H_i$.  Maximizing over $H_i\in[0,1]$ then gives the worst-case gap.  Straightforward calculus yields the saddle-point:
\begin{align*}
\max_{H_i\in[0,1]} & \inf_{\alpha_i > 0} \left\{ \left(\frac{1}{2}-\frac{1}{\alpha_i} \right) H_i +\frac{1}{\alpha_i} \log(1+2^{-\alpha_i})\right\} \\
&=   \inf_{\alpha_i > 0} \max_{H_i\in[0,1]} \left\{ \left(\frac{1}{2}-\frac{1}{\alpha_i} \right) H_i +\frac{1}{\alpha_i} \log(1+2^{-\alpha_i})\right\}\\
&=\frac{1}{2}\log\left(\frac{5}{4}\right) < 0.161,
\end{align*}
which is achieved for $\alpha_i=2$ and any $H\in[0,1]$.  
\end{proof}
\begin{remark}
We note briefly that this estimate can potentially be improved if one knows more about the source distribution.
\end{remark}

\section{Concluding Remarks} \label{sec:Conc}
One immediate direction for further work would be to extend our results to more than two encoders.  For the CEO problem, our results can be extended to an arbitrary number of encoders.  This extension is proved in Appendix \ref{app:mEncCEO}.  

On the other hand, generalizing the results for the two-encoder source coding problem with distortion constraints on $Y_1$ and $Y_2$ poses a significant challenge.  The obvious point of difficulty in the proof is extending the interpolation argument to higher dimensions so that it yields a distribution with the desired properties.  In fact, a ``quick-fix" to the interpolation argument alone would not be sufficient since this would imply that the Berger-Tung inner bound is tight for more than two encoders.  This is known to be false (even for the logarithmic loss distortion measure) since the Berger-Tung achievability scheme is not optimal for the lossless modulo-sum problem studied by K\"{o}rner and Marton in \cite{bib:KornerMarton1979}. 
%Therefore, our results suggest that the solution to a multiterminal problem with $m$ encoders will not necessarily yield a clear solution path for the same problem with $m+1$ encoders.  

\section*{Acknowledgement}
The authors would like to thank Professors Suhas Diggavi and Aaron Wagner for the helpful discussions on this topic.  
%The author is deeply indebted to Professor Tsachy Weissman for the encouragement and suggestions he provided throughout the preparation of this manuscript.

\appendix

\section{Cardinality Bounds on Auxiliary Random Variables} \label{app:cardBounds}

In order to obtain tight cardinality bounds on the auxiliary random variables used throughout this paper, we refer to a recent result by Jana.  In \cite{bib:Jana2009}, the author carefully applies the Caratheodory-Fenchel-Eggleston theorem in order to obtain tight cardinality bounds on the auxiliary random variables in the Berger-Tung inner bound.  This result extends the results and techniques employed by Gu and Effros for the Wyner-Ahlswede-K\"{o}rner problem \cite{bib:GuEffros2007}, and by Gu, Jana, and Effros for the Wyner-Ziv problem \cite{bib:GuJana2008}.  We now state Jana's result, appropriately modified for our purposes:

Consider an arbitrary joint distribution $p(v,y_1,\dots,y_m)$ with random variables $V,Y_1,\dots,Y_m$ coming from alphabets $\mathcal{V},\mathcal{Y}_1,\dots, \mathcal{Y}_m$ respectively.

Let $d_l: \mathcal{V}\times \hat{\mathcal{V}}_l \rightarrow \mathbb{R}$, $1\leq l \leq L$ be arbitrary distortion measures defined for possibly different reproduction alphabets $\hat{\mathcal{V}}_l$.

\begin{definition}
Define $\mathcal{A}^{\star}$ to be the set of $(m+L)$-vectors $(R_1,\dots,R_m,D_1,\dots,D_L)$ satisfying the following conditions:
\begin{enumerate}
\item auxiliary random variables $U_1,\dots,U_m$ exist such that 
\begin{align*}
\sum_{i\in \mathcal{I}}R_i \geq I(Y_{\mathcal{I}};U_{\mathcal{I}}|U_{\mathcal{I}^c}), \mbox{~for all  $\mathcal{I}\subseteq \{1,\dots,m\}$, and}
\end{align*}
\item mappings $\psi_l : \mathcal{U}_1\times \dots \times \mathcal{U}_m \rightarrow \hat{\mathcal{V}}_l$, $1\leq l \leq L$ exist such that
\begin{align*}
\mathbb{E} d_l(V,\psi_l(U_1,\dots, U_m))\leq D_l
\end{align*}
\end{enumerate}
for some joint distribution
\begin{align*}
p(v,y_1,\dots,y_m)\prod_{j=1}^m p(u_j|y_j).
\end{align*}
\end{definition}

\begin{lemma}[Lemma 2.2 from \cite{bib:Jana2009}]\label{lem:Jana}
Every extreme point of $\mathcal{A}^{\star}$ corresponds to
some choice of auxiliary variables $U_1,\dots, U_m$ with alphabet
sizes $|\mathcal{U}_j|\leq |\mathcal{Y}_j|$, $1\leq j \leq m$.
\end{lemma}

In order to obtain the cardinality bounds for the CEO problem, we simply let $L=1$, $V=X$, and $\hat{\mathcal{V}}_1=\hat{\mathcal{X}}$.  Defining 
\begin{align*}
d_1({x},\hat{x})=\log\left(\frac{1}{\hat{x}(x)} \right),
\end{align*}
we see that $\overline{ \mathcal{RD} }_{CEO}^{ \star } = \mbox{conv} \left(\mathcal{A}^{\star} \right)$, where $\mbox{conv} \left(\mathcal{A}^{\star} \right)$ denotes the convex hull of $\mathcal{A}^{\star}$.  Therefore, Lemma \ref{lem:Jana} implies that all extreme points of $\overline{ \mathcal{RD} }_{CEO}^{ \star }$ are achieved with a choice of auxiliary random variables $U_1,\dots, U_m$ with alphabet
sizes $|\mathcal{U}_j|\leq |\mathcal{Y}_j|$, $1\leq j \leq m$. By timesharing between extreme points, any point in $\overline{ \mathcal{RD} }_{CEO}^{ \star }$ can be achieved for these alphabet sizes.

Obtaining the cardinality bounds for the multiterminal source coding problem proceeds in a similar fashion.  In particular, let $L=m=2$, $V=(Y_1,Y_2)$, and $\hat{\mathcal{V}}_j=\hat{\mathcal{Y}_j}$, $j=1,2$.  Defining 
\begin{align*}
d_j((y_1,y_2),\hat{y}_j)=\log\left(\frac{1}{\hat{y}_j(y_j)} \right) \mbox{~for $j=1,2$},
\end{align*}
we see that $\overline{ \mathcal{RD} }^{ \star } = \mbox{conv} \left(\mathcal{A}^{\star} \right)$.  In this case, Lemma \ref{lem:Jana} implies that all extreme points of $\overline{ \mathcal{RD} }^{ \star }$ are achieved with a choice of auxiliary random variables $U_1,U_2$ with alphabet
sizes $|\mathcal{U}_j|\leq |\mathcal{Y}_j|$, $1\leq j \leq 2$.  By timesharing between extreme points, any point in $\overline{ \mathcal{RD} }^{ \star }$ can be achieved for these alphabet sizes.

In order to obtain cardinality bounds on the timesharing variable $Q$, we can apply Caratheodory's theorem (cf. \cite{bib:Witsenhausen1980}). In particular, if $C\subset\mathbb{R}^n$ is compact, then any point in $\mbox{conv}(C)$ is a convex combination of at most $n+1$ points of $C$.  Taking $C$ to be the closure of the set of extreme points of $\mathcal{A}^{\star}$ is sufficient for our purposes (boundedness of $C$ can be dealt with by a standard truncation argument).

\section{Extension of CEO Results to $m$ Encoders} \label{app:mEncCEO}
In this appendix, we prove the generalization of Theorem \ref{thm:CEOregion} to $m$ encoders, which essentially amounts to extending the argument in the proof of Theorem \ref{thm:CEOregion} to the general case.  We begin by stating the $m$-encoder generalizations of Theorems \ref{thm:CEOAchregion} and \ref{thm:CEOconv}, the proofs of which are trivial extensions of the proofs given for the two-encoder case and are therefore omitted.

\begin{definition}
Let $\mathcal{R}^i_{CEO,m}$ be the set of all $(R_1,\dots,R_m,D)$ satisfying
\begin{align*}
\sum_{i\in \mathcal{I}} R_i & \geq I(Y_{\mathcal{I}};U_{\mathcal{I}}|U_{\mathcal{I}^c},Q) \mbox{~for all $\mathcal{I}\subseteq \{1,\dots,m\}$}\\
D &\geq H(X|U_1,\dots ,U_m ,Q).
\end{align*}
for some joint distribution $p(q)p(x)\prod_{i=1}^m p(y_i|x) p(u_i|y_i,q)$. % with $|\mathcal{U}_j| \leq |\mathcal{Y}_j|+4$ and $|\mathcal{Q}|\leq 5$.
\end{definition}

\begin{theorem} \label{thm:CEOAchregionMenc}
All rate distortion vectors $(R_1,\dots,R_m,D)\in \mathcal{R}^i_{CEO,m}$ are achievable.
\end{theorem}

\begin{definition}
Let $\mathcal{R}^o_{CEO,m}$ be the set of $(R_1,\dots,R_m,D)$ satisfying
\begin{align}
\sum_{i\in \mathcal{I}} R_i & \geq \sum_{i\in \mathcal{I}} I(U_i;Y_i|X,Q) +
H(X|U_{\mathcal{I}^c},Q)-D  \mbox{~for all $\mathcal{I}\subseteq \{1,\dots,m\}$} \label{eqn:Ro1} \\
D &\geq H(X|U_1,\dots ,U_m ,Q) \label{eqn:Ro2}.
\end{align}
for some joint distribution $p(q)p(x)\prod_{i=1}^m p(y_i|x) p(u_i|y_i,q)$. % with $|\mathcal{U}_j| \leq |\mathcal{Y}_j|+4$ and $|\mathcal{Q}|\leq 5$.
\end{definition}

\begin{theorem} \label{thm:CEOconvMenc}
If $(R_1,\dots,R_m,D)$ is strict-sense achievable, then  \linebreak $(R_1,\dots,R_m,D) \in \mathcal{R}^o_{CEO,m}$.
\end{theorem}

Given the definitions of $\mathcal{R}^i_{CEO,m}$ and $\mathcal{R}^o_{CEO,m}$, the generalization of Theorem \ref{thm:CEOregion} to $m$ encoders is an immediate consequence of the following lemma:
\begin{lemma} \label{lem:dominationLemmaMenc}
$\mathcal{R}^o_{CEO,m} \subseteq \mathcal{R}^i_{CEO,m}$.
\end{lemma}
\begin{proof}
Suppose $(R_1,\dots,R_m,D)\in \mathcal{R}^o_{CEO,m}$, then by definition there exists $p(q)$ and  conditional distributions $\{p(u_i|y_i,q)\}_{i=1}^m$ so that \eqref{eqn:Ro1} and \eqref{eqn:Ro2} are satisfied.  For the joint distribution corresponding to $p(q)$ and  conditional distributions $p\{(u_i|y_i,q)\}_{i=1}^m$, define $\mathcal{P}_D\subset \mathbb{R}^m$ to be the polytope defined by the inequalities \eqref{eqn:Ro1}.  Now, to show $(R_1,\dots,R_m,D)\in \mathcal{R}^i_{CEO,m}$, it suffices to show that each extreme point of $\mathcal{P}_D$ is dominated by a point in $\mathcal{R}^i_{CEO,m}$ that achieves distortion at most $D$.

%Define $\mathcal{R}_o^{(m)}(D)$ to be the set of $(R_1,\dots,R_m)$ such that $(R_1,\dots,R_m,D)\in \mathcal{R}_o^{(m)}$.  As in the proof of Lemma \ref{lem:dominationLemma}, we analyze the extreme points of $\mathcal{R}_o^{(m)}$ and show that each is dominated by a point in $\mathcal{R}_i^{(m)}$.  Specifically, it suffices to look at the extreme points of  $\mathcal{R}_o^{(m)}(D)$ when the joint distribution is fixed.  To this end, fix $p(q)$, the conditional distributions $p(u_i|y_i,q)$, and the distortion $D$ (which must satisfy $D\geq H(X|U_1,\dots,U_m,Q)$). 
%this ^^ should be modified to make it more clear.

To this end, define the set function $f: 2^{[m]} \rightarrow \mathbb{R}$ as follows:
\begin{align}
f(\mathcal{I}) &:= I(Y_{\mathcal{I}};U_{\mathcal{I}}| U_{\mathcal{I}^c},Q)-(D-H(X|U_1,\dots,U_m,Q)) \notag \\
&=\sum_{i\in \mathcal{I}} I(U_i;Y_i|X,Q) +
H(X|U_{\mathcal{I}^c},Q)-D. \notag
\end{align}

It can be verified that the function $f$ and the function $f^+(\mathcal{I})=\max \{ f(\mathcal{I}),0 \}$ are supermodular functions (see Appendix \ref{app:submodular}).  By construction, $\mathcal{P}_D$ is equal to the set of $(R_1,\dots,R_m)$ which satisfy:
\begin{align*}
\sum_{i\in \mathcal{I}} R_i &\geq f^+(\mathcal{I}).
\end{align*}

It follows by basic results in submodular optimization (see Appendix \ref{app:submodular}) that, for a linear ordering $i_1\prec i_2 \prec \dots \prec i_m$ of $\{1,\dots,m\}$, an extreme point of $\mathcal{P}_D$ can be greedily computed as follows:
\begin{align*}
\tilde{R}_{i_j} = f^+(\{i_1,\dots,i_j \})-f^+(\{i_1,\dots,i_{j-1} \}) \mbox{~for $j=1,\dots,m$}.
\end{align*}
Furthermore, all extreme points of  $\mathcal{P}_D$ can be enumerated by looking over all linear orderings $i_1\prec i_2 \prec \dots \prec i_m$ of $\{1,\dots,m\}$.  Each ordering of $\{1,\dots,m\}$ is analyzed in the same manner, hence we assume (for notational simplicity) that the ordering we consider is the natural ordering $i_j=j$.

Let $j$ be the first index for which $\tilde{R}_j>0$.  Then, by construction,
\begin{align*}
\tilde{R}_k=I(U_k;Y_k|U_{k+1},\dots , U_m,Q) \mbox{~for all $k>j$.}
\end{align*}
Furthermore, we must have $f(\{1,\dots,j'\}) \leq 0$ for all $j'<j$.  Thus, $\tilde{R}_j$ can be expressed as
\begin{align*}
\tilde{R}_j &= \sum_{i=1}^j I(Y_i;U_i|X,Q)+H(X|U_{j+1},\dots,U_m,Q)-D \\
&=I(Y_j;U_j|U_{j+1},\dots,U_m,Q) + f(\{1,\dots,j-1\})\\
&=(1-\theta)I(Y_j;U_j|U_{j+1},\dots,U_m,Q),
\end{align*}
where $\theta \in [0,1)$ is defined as:
\begin{align*}
\theta&=\frac{-f(\{1,\dots,j-1\})}{I(Y_j;U_j|U_{j+1},\dots,U_m,Q)}\\
&=\frac{D-H(X|U_1,\dots,U_m,Q)-I(U_1,\dots,U_{j-1};Y_1,\dots,Y_{j-1}|U_j,\dots,U_m,Q)}{I(Y_j;U_j|U_{j+1},\dots,U_m,Q)}.
\end{align*}

By the results of Theorem \ref{thm:CEOAchregionMenc}, the rates $(\tilde{R}_1,\dots,\tilde{R}_m)$ permit the following coding scheme:  For a fraction $(1-\theta)$ of the time, a codebook can be used that allows the decoder to recover $U_j^n,\dots,U_m^n$ with high probability.  The other fraction $\theta$ of the time, a codebook can be used that allows the decoder to recover $U_{j+1}^n,\dots,U_m^n$ with high probability.  As $n\rightarrow \infty$, this coding scheme can achieve distortion
%
%Observing that $(\tilde{R}_1,\dots,\tilde{R}_m)$ satisfy the inequalities defining $\mathcal{R}_i^{(m)}$ for the joint distribution corresponding to $p(q)$ and  conditional distributions $p(u_i|y_i,q)$, Theorem \ref{thm:CEOAchregionMenc} implies that, for the choice of rates $(\tilde{R}_1,\dots,\tilde{R}_m)$, the distortion
\begin{align}
\tilde{D}&=(1-\theta)H(X|U_j,\dots,U_m,Q)+\theta H(X|U_{j+1},\dots,U_m,Q)\notag\\
&=H(X|U_j,\dots,U_m,Q)+\theta I(X;U_j|U_{j+1},\dots,U_m,Q)\notag\\
&=H(X|U_j,\dots,U_m,Q)+\frac{I(X;U_j|U_{j+1},\dots,U_m,Q)}{I(Y_j;U_j|U_{j+1},\dots,U_m,Q)}\times \notag\\
& \quad \left[D-H(X|U_1,\dots,U_m,Q)-I(U_1,\dots,U_{j-1};Y_1,\dots,Y_{j-1}|U_j,\dots,U_m,Q)\right]\notag\\
&\leq H(X|U_j,\dots,U_m,Q)+D-H(X|U_1,\dots,U_m,Q)\notag\\
&\quad -I(U_1,\dots,U_{j-1};Y_1,\dots,Y_{j-1}|U_j,\dots,U_m,Q) \label{eqn:DPI} \\
&= D+I(X;U_1,\dots U_{j-1}| U_j,\dots,U_m,Q) \notag \\
&\quad -I(U_1,\dots,U_{j-1};Y_1,\dots,Y_{j-1}|U_j,\dots,U_m,Q)\notag\\
&= D-I(U_1,\dots,U_{j-1};Y_1,\dots,Y_{j-1}|X,U_j,\dots,U_m,Q)\notag\\
&\leq D.\label{eqn:nonNegI}
\end{align}
In the preceding string of inequalities \eqref{eqn:DPI} follows since $U_j$ is conditionally independent of everything else given $(Y_j,Q)$, and \eqref{eqn:nonNegI} follows from the non-negativity of mutual information.

Therefore, for every extreme point $(\tilde{R}_1,\dots,\tilde{R}_m)$ of $\mathcal{P}_D$, the point \linebreak $(\tilde{R}_1,\dots,\tilde{R}_m,D)$ lies in $\mathcal{R}^i_{CEO,m}$. This proves the lemma.
\end{proof}

Finally, we remark that the results of Appendix \ref{app:cardBounds} imply that it suffices to consider auxiliary random variables $U_1,\dots, U_m$ with alphabet
sizes $|\mathcal{U}_j|\leq |\mathcal{Y}_j|$, $1\leq j \leq m$.  The timesharing variable $Q$ requires an alphabet size bounded by $|\mathcal{Q}|\leq m+2$.

\section{Supermodular Functions} \label{app:submodular}
In this appendix, we review some basic results in submodular optimization that were used in Appendix \ref{app:mEncCEO} to prove Lemma \ref{lem:dominationLemmaMenc}.  We tailor our statements toward supermodularity, since this is the property we require in Appendix \ref{app:mEncCEO}.

We begin by defining a supermodular function.  
\begin{definition}
Let $E=\{1,\dots,n\}$ be a finite set.  A function $s:2^E\rightarrow \mathbb{R}$ is \emph{supermodular} if for all $S,T\subseteq  E$
\begin{align}
s(S)+s(T)\leq s(S\cap T)+s(S \cup T). \label{eqn:submodDefn}
\end{align}
\end{definition}

One of the fundamental results in submodular optimization is that a greedy algorithm minimizes a linear function over a supermodular polyhedron.  By varying the linear function to be minimized, all extreme points of the supermodular polyhedron can be enumerated.  In particular, define the supermodular polyhedron $\mathcal{P}(s)\subset \mathbb{R}^n$ to be the set of $x\in \mathbb{R}^n$ satisfying
\begin{align*}
\sum_{i\in T} x_i \geq s(T) \mbox{~for all $T\subseteq E$}.
\end{align*}
The following theorem provides an algorithm that enumerates the extreme points of $\mathcal{P}(s)$.

\begin{theorem}[See \cite{bib:Schrijver2003,bib:Fujishige2010,bib:McCormick2005}]\label{thm:greedy}
For a linear ordering $e_1 \prec e_2 \prec \dots \prec e_n$ of the elements in $E$, Algorithm \ref{alg:greedy} returns an extreme point $v$ of $\mathcal{P}(s)$.  Moreover, all extreme points of $\mathcal{P}(s)$ can be enumerated by considering all linear orderings of the elements of $E$.

\begin{pseudocode}[ruled]{Greedy}{s,E,\prec} \label{alg:greedy}
\COMMENT{ Returns extreme point $v$ of $\mathcal{P}(s)$ corresponding to the ordering $\prec$. }\\

\FOR i = 1,\dots n \\
\quad \mbox{Set~} v_i = s(\{e_1,e_2,\dots,e_i\})-s(\{e_1,e_2,\dots,e_{i-1}\})\\
\RETURN{v}
\end{pseudocode}
\end{theorem}
\begin{proof}
See  \cite{bib:Schrijver2003,bib:Fujishige2010,bib:McCormick2005}.
\end{proof}

Theorem \ref{thm:greedy} is the key tool we employ to establish Lemma \ref{lem:dominationLemmaMenc}.  In order to apply it, we require the following lemma. 

\begin{lemma}
For any joint distribution of the form $p(q)p(x)\prod_{i=1}^m p(y_i|x) p(u_i|y_i,q)$ and fixed $D\in \mathbb{R}$, define the set function $f: 2^{[m]} \rightarrow \mathbb{R}$ as:
\begin{align}
f(\mathcal{I}) &:= I(Y_{\mathcal{I}};U_{\mathcal{I}}| U_{\mathcal{I}^c},Q)-(D-H(X|U_1,\dots,U_m,Q)) \label{eqn:setFunction2} \\
&=\sum_{i\in \mathcal{I}} I(U_i;Y_i|X,Q) +
H(X|U_{\mathcal{I}^c},Q)-D, \notag
\end{align}
and the corresponding non-negative set function $f^+: 2^{[m]} \rightarrow \mathbb{R}$ as $f^+=\max\{ f,0\}$. The functions $f$ and $f^+$ are supermodular.
\end{lemma}
\begin{proof}
In order to verify that $f$ is supermodular, it suffices to check that the function $f'(\mathcal{I})=I(Y_{\mathcal{I}};U_{\mathcal{I}}|U_{\mathcal{I}^c},Q)$ is supermodular since the latter two terms in \eqref{eqn:setFunction2} are constant.  To this end, consider sets $T,S \subseteq \{1,\dots,m\}$ and observe that:
\begin{align}
f'(S)+f'(T)&=I(Y_{{S}};U_{{S}}|U_{{S}^c},Q)+I(Y_{{T}};U_{{T}}|U_{{T}^c},Q)\notag \\
&=H(U_{S}|U_{S^c},Q)-H(U_{S}|Y_{S},Q)+H(U_{T}|U_{T^c},Q)-H(U_{T}|Y_{T},Q)\notag\\
&=H(U_{S}|U_{S^c},Q)+H(U_{T}|U_{T^c},Q) \notag \\
&\quad -H(U_{S\cup T}|Y_{S\cup T},Q) - H(U_{S\cap T}|Y_{S\cap T},Q)\label{eqn:condIndepSM} \\
&=H(U_{S\backslash T}|U_{S^c},Q)+H(U_{S \cap T}|U_{(S \cap T)^c},Q) + H(U_{T}|U_{T^c},Q)\notag\\
&\quad -H(U_{S\cup T}|Y_{S\cup T},Q) - H(U_{S\cap T}|Y_{S\cap T},Q) \label{eqn:chainRuleSM} \\
&=H(U_{S\backslash T}|U_{S^c},Q)+H(U_{T}|U_{T^c},Q)-H(U_{S\cup T}|Y_{S\cup T},Q)\notag \\
&\quad + I(U_{S \cap T};Y_{S \cap T} |U_{(S \cap T)^c},Q)\notag \\
&\leq H(U_{S\backslash T}|U_{(S\cup T)^c},Q)+H(U_{T}|U_{T^c},Q)-H(U_{S\cup T}|Y_{S\cup T},Q) \notag\\
&\quad + I(U_{S \cap T};Y_{S \cap T} |U_{(S \cap T)^c},Q) \label{eqn:condRedSM} \\
&= I(U_{S\cup T};Y_{S\cup T}|U_{(S\cup T)^c},Q) + I(U_{S \cap T};Y_{S \cap T} |U_{(S \cap T)^c},Q)\notag\\
&=f'(S \cap T)+f'(S \cup T).\notag
\end{align}
The labeled steps above can be justified as follows:
\begin{itemize}
\item \eqref{eqn:condIndepSM} follows since $U_i$ is conditionally independent of everything else given $(Y_i,Q)$.
\item \eqref{eqn:chainRuleSM} is simply the chain rule.
\item \eqref{eqn:condRedSM} follows since conditioning reduces entropy.
\end{itemize}

Next, we show that $f^+=\max\{f,0\}$ is supermodular.  Observe first that $f$ is monotone increasing, i.e., if $S\subset T$, then $f(S) \leq f(T)$. Thus, fixing $S,T\subseteq \{1,\dots,m\}$, we can assume without loss of generality that 
\begin{align*}
f(S \cap T)\leq f(S )\leq f(T)\leq f(S \cup T).
\end{align*}
If $f(S \cap T)\geq 0$, then \eqref{eqn:submodDefn} is satisfied for $s=f^+$ by the supermodularity of $f$.  On the other hand, if $f(S \cup T)\leq 0$, then \eqref{eqn:submodDefn} is a tautology for $s=f^+$.  Therefore, it suffices to check the following three cases:
\begin{itemize}
\item Case 1: $f(S \cap T)\leq 0 \leq f(S )\leq f(T)\leq f(S \cup T)$.  In this case, the supermodularity of $f$ and the fact that $f^+\geq f$ imply:
\begin{align*}
f^+(S \cup T) + f^+(S \cap T) &\geq f(S \cup T)+f(S \cap T)\\ &\geq f(S)+f(T) = f^+(S)+f^+(T).
\end{align*}
\item Case 2: $f(S \cap T)\leq  f(S )\leq 0 \leq f(T)\leq f(S \cup T)$.  Since $f$ is monotone increasing, we have:
\begin{align*}
f^+(S \cup T) + f^+(S \cap T) = f(S \cup T)+0 \geq f(T) +0 = f^+(S)+f^+(T).
\end{align*}
\item Case 3: $f(S \cap T)\leq  f(S )\leq f(T)\leq 0 \leq f(S \cup T)$.  By definition of $f^+$:
\begin{align*}
f^+(S \cup T) + f^+(S \cap T) = f(S \cup T)+0 \geq 0 +0 = f^+(S)+f^+(T).
\end{align*}
\end{itemize}
Hence, $f^+ =\max\{f,0\}$ is supermodular.

\end{proof}

\section{Amplifying a Pointwise Convexity Constraint}\label{app:convexityLemma}

\begin{lemma}\label{lem:amplifyConvex}
Let $r_1,r_2\in\mathbb{R}$ be given, and suppose $f_1 : K\rightarrow \mathbb{R}$ and $f_2 : K\rightarrow \mathbb{R}$ are continuous functions defined on a compact domain $K\subset \mathbb{R}^n$. If there exists a function $h:[0,1]\rightarrow K$ satisfying
\begin{align}
t \left( f_1\circ h\right) (t)+(1-t)\left( f_2\circ h\right) (t) \leq t r_1 +(1-t) r_2 \mbox{~~for all $t\in[0,1]$,} \label{eqn:cond1}
\end{align}
then there exists $x_1^*, x_2^*\in K$ and $t^*\in [0,1]$ for which
\begin{align*}
t^{*} f_1( x_{1}^{*} ) + (1-t^{*}) f_1( x_{2}^{*})&\leq r_1 \\
t^{*} f_2( x_{1}^{*} ) + (1-t^{*}) f_2( x_{2}^{*})&\leq r_2.
\end{align*}
\end{lemma}

Before we prove the lemma, we make a few remarks.
At first glance, this lemma appears somewhat bizarre. Indeed, the set $K$ need only be compact (e.g., connectedness is not required) and $h$ can be an arbitrarily complicated function, as long as it satisfies \eqref{eqn:cond1}.  The strange nature of the lemma is echoed by the proof in that we merely prove the existence of the desired $x_1^*$, $ x_2^*$ and $t^*$; no further information is obtained.  Stripped to its core, the existence of the desired $x_1^*$, $ x_2^*$ and $t^*$ essentially follows from the pigeonhole principle, which manifests itself in the sequential compactness of $K$.

Despite its strange nature, Lemma \ref{lem:amplifyConvex} is crucial in establishing the converse result for the multiterminal source coding problem under logarithmic loss.  In this application, $K$ is taken to be a closed subset of a finite-dimensional probability simplex and $f_1, f_2$ are conditional entropies evaluated for probability distributions in $K$.

Finally, we remark that the Lemma \ref{lem:amplifyConvex} can be generalized to a certain extent.  For example, the function $h$ need only be defined on a dense subset of $[0,1]$ and the set $K$ can be a more general sequentially compact space.

\begin{proof}[Proof of Lemma \ref{lem:amplifyConvex}]
Since $f_1,f_2$ are continuous\footnote{Although not required for our purposes, we can assume $f_1$ and $f_2$ are defined and continuous over all of $\mathbb{R}^n$.  This is a consequence of the Tietze extension theorem.} and $K$ is compact, there exists $M<\infty$ such that $f_1$ and $f_2$ are bounded from above and below by $M$ and $-M$, respectively.
Fix $\epsilon>0$, and partition the interval $[0,1]$ as $0=t_1<t_2<\dots<t_m=1$, such that $|t_{j+1}-t_j|<\frac{\epsilon}{M}$.    For convenience define $x_{t_{j}}:=h(t_j)$ when $t_j$ is in the partition.

Now, for $i=1,2$ define piecewise-linear functions $g_1(t),g_2(t)$ on [0,1] by:
\begin{align}
g_i(t)=\left\{\begin{array}{ll}
f_i(x_{t_j}) & \mbox{if $t_j$ is in the partition} \\
\theta f_i( x_{t_j} ) + (1-\theta) f_i( x_{t_{j+1}}) & \mbox{if $t$ is in the interval $(t_j,t_{j+1})$,}
\end{array} \right.
\end{align}
where $\theta\in (0,1)$ is chosen so that $t=\theta t_j+(1-\theta)t_{j+1}$ when $t$ is in the interval $(t_j,t_{j+1})$.

With $g_1(t)$ and $g_2(t)$ defined in this manner, suppose $t=\theta t_j+(1-\theta)t_{j+1}$ for some $j$ and $\theta$.  Then some straightforward algebra yields:
\begin{align}
t g_1(t)+(1-t)g_2(t) &= (\theta t_j+(1-\theta)t_{j+1})\left( \theta f_1(x_{t_j})+ (1-\theta) f_1(x_{t_{j+1}}) \right) \notag\\
&\quad +(1-\theta t_j-(1-\theta)t_{j+1})\left( \theta f_2(x_{t_j})+ (1-\theta) f_2(x_{t_{j+1}}) \right) \notag\\
&=\theta^2\left[t_j f_1(x_{t_j})+ (1-t_j) f_2(x_{t_j})\right] \notag\\
&\quad +(1-\theta)^2\left[t_{j+1} f_1(x_{t_{j+1}})+ (1-t_{j+1}) f_2(x_{t_{j+1}})\right] \notag\\
&\quad +\theta(1-\theta)\left[(1-t_j)f_2(x_{t_{j+1}})+ (1-t_{j+1})f_2(x_{t_{j}})  \right. \notag\\
&\quad\quad\quad\quad\quad\quad \left. +t_{j+1} f_1(x_{t_{j}})+ t_{j} f_1(x_{t_{j+1}})\right] \notag\\
&\leq \theta^2\left[t_j f_1(x_{t_{j}})+ (1-t_j) f_2(x_{t_{j}})\right] \notag\\
&\quad +(1-\theta)^2\left[t_{j+1} f_1(x_{t_{j+1}})+ (1-t_{j+1}) f_2(x_{t_{j+1}})\right]\notag\\
&\quad +\theta(1-\theta)\left[(1-t_{j+1})f_2(x_{t_{j+1}})+ (1-t_{j})f_2(x_{t_{j}}) \right. \notag\\
&\quad\quad\quad\quad\quad\quad \left. + t_{j} f_1(x_{t_{j}})+ t_{j+1} f_1(x_{t_{j+1}})\right] +\epsilon \notag\\
&\leq \theta^2\left[t_j r_1+ (1-t_j) r_2\right] \notag\\
&\quad +(1-\theta)^2\left[t_{j+1} r_1+ (1-t_{j+1}) r_2\right]\notag\\
&\quad +\theta(1-\theta)\left[(1-t_{j+1}) r_2+ (1-t_{j}) r_2 \right. \notag\\
&\quad\quad\quad\quad\quad\quad \left. + t_{j} r_1+ t_{j+1} r_1\right] +\epsilon \notag\\
&=(\theta t_j+(1-\theta)t_{j+1})r_1 + (1-\theta t_j-(1-\theta)t_{j+1})r_2+\epsilon\notag \\
&= t r_1+(1-t)r_2+\epsilon, \label{eqn:ineqToContradict}
\end{align}
\begin{figure}[t!]
\begin{center}
\def\svgwidth{4.5in}
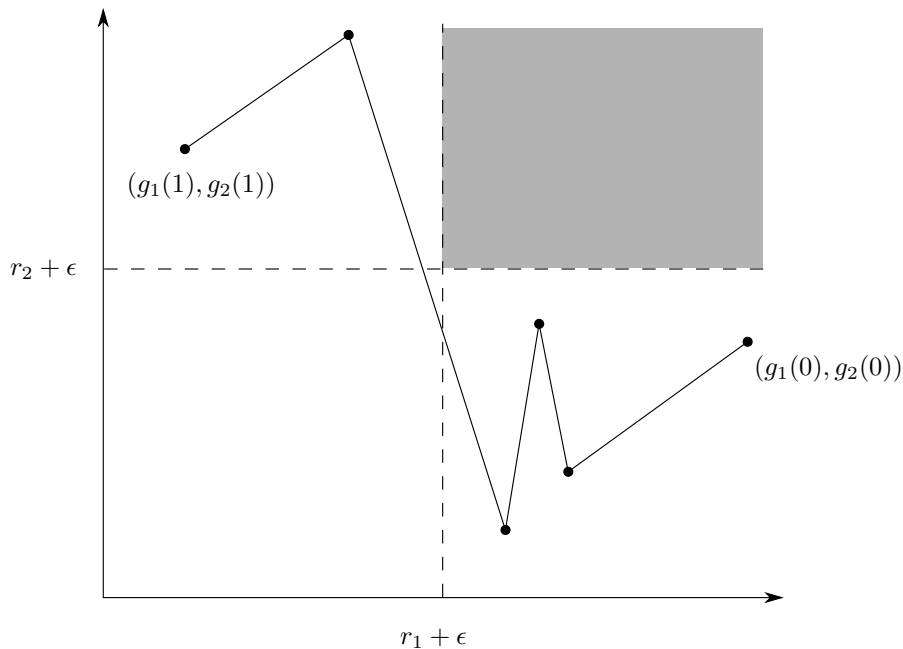
\caption[Illustration of the key step in the proof of Lemma \ref{lem:amplifyConvex}.]{A parametric plot of the function $\varphi : t \mapsto (g_1(t),g_2(t))$.  Since $\varphi(t)$ is continuous, starts with $g_2(0)\leq r_2+\epsilon$, ends with $g_1(1)\leq r_1+\epsilon$, and doesn't intersect the shaded area, $\varphi(t)$ must pass through the lower-left region.}\label{fig:g1_g2}
%\includegraphics[trim = 30mm 73mm 30mm 90mm, clip, scale=.8]{rd_pigeonhole_niceFig.pdf}
%\caption[Illustration of the key step in the proof of Lemma \ref{lem:amplifyConvex}.]{An example of piecewise-linear functions $g_1(t)$ and $g_2(t)$ satisfying $tg_1(t)+(1-t)g_2(t)\leq tr_1 + (1-t)r_2$ whenever $t\in\{0,1/10,\dots,9/10, 1\}$.  The vertical bar indicates the location of a $t^*$ where $g_1(t^*)\leq r_1$ and $g_2(t^*)\leq r_2$.}\label{fig:g1_g2}
\end{center}
\end{figure}
where the first inequality follows since $|t_{j+1}-t_j|$ is small, and the second inequality follows from the the fact that \eqref{eqn:cond1} holds for each $t_j$ in the partition.  Notably, this implies that it is impossible to have 
\begin{align*}
g_1(t)> r_1+\epsilon \mbox{~~and~~}g_2(t)> r_2+\epsilon
\end{align*}
hold simultaneously for any $t\in[0,1]$, else we would obtain a contradiction to \eqref{eqn:ineqToContradict}.
Also, since we included the endpoints $t_1=0$ and $t_m=1$ in the partition, we have the following two inequalities:
\begin{align*}
g_1(1)\leq r_1, \mbox{~and~}
g_2(0)\leq r_2.
\end{align*}

Combining these observations with the fact that $g_1(t)$ and $g_2(t)$ are continuous, there must exist some $t^*\in[0,1]$ for which 
\begin{align*}
g_1(t^*)\leq r_1+\epsilon , \mbox{~~and~~}
g_2(t^*)\leq r_2+\epsilon
\end{align*}
simultaneously. An illustration of this is given in Figure \ref{fig:g1_g2}, which is a mere variation on the classical intermediate value theorem.

Applying this result, we can find a sequence $\{x_{1}^{(n)},x_{2}^{(n)},t^{(n)}\}_{n=1}^\infty$ satisfying
\begin{align*}
t^{(n)} f_1( x_{1}^{(n)} ) + (1-t^{(n)}) f_1( x_{2}^{(n)})&\leq r_1+\frac{1}{n} \\
t^{(n)} f_2( x_{1}^{(n)} ) + (1-t^{(n)}) f_2( x_{2}^{(n)})&\leq r_2+\frac{1}{n}
\end{align*}
for each $n\geq 1$.  Since $K \times K \times [0,1]$ is sequentially compact, there exists a convergent subsequence $\{n_j\}_{j=1}^{\infty}$ such that 
$(x_{1}^{(n_j)},x_{2}^{(n_j)},t^{(n_j)})\rightarrow (x_{1}^*,x_{2}^*,t^*) \in K\times K \times [0,1]$.  The continuity of $f_1$ and $f_2$ then apply to yield the desired result.
\end{proof}

\section{Strengthening the Converse of Theorem \ref{thm:MTSCregionDesc}}\label{app:strengthenedConverse}

In this appendix, we prove a stronger version of the converse of Theorem \ref{thm:MTSCregionDesc}.
To be precise, let $\hat{\mathcal{Y}}_1^{*n}$ and $\hat{\mathcal{Y}}_2^{*n}$ denote the set of probability measures on $\mathcal{Y}_1^n$ and $\mathcal{Y}_2^n$, respectively.  Let $d_1^*,d_2^*$ be the (extended)-log loss distortion measures defined as follows:
\begin{align*}
d_1^*(y_1^n,\hat{y}_1^n) &= \frac{1}{n}\log\left(\frac{1}{\hat{y}_1^n(y_1^n)} \right)\\
d_2^*(y_2^n,\hat{y}_2^n) &= \frac{1}{n}\log\left(\frac{1}{\hat{y}_2^n(y_2^n)} \right),
\end{align*}
where $\hat{y}_1^n(y_1^n)$ is the probability assigned to outcome $y_1^n\in \mathcal{Y}_1^n$ by the probability measure $\hat{y}_1^n\in \hat{\mathcal{Y}}_1^{*n}$.  Similarly for $\hat{y}_2^n(y_2^n)$.  Note that this extends the standard definition of logarithmic loss to sequence reproductions.

\begin{definition}
We say that a tuple $(R_1,R_2,D_1,D_2)$ is sequence-achievable if, for any $\epsilon>0$, there exist encoding functions 
\begin{align*}
f_1&:\mathcal{Y}_1^n \rightarrow \{1,\dots,2^{nR_1}\}\\
f_2 &:\mathcal{Y}_2^n \rightarrow \{1,\dots,2^{nR_2}\},
\end{align*}
and decoding functions 
\begin{align*}
\phi_1&: \{1,\dots,2^{nR_1}\}\times \{1,\dots,2^{nR_2}\} \rightarrow \hat{\mathcal{Y}}_1^{*n}\\
\phi_2 &: \{1,\dots,2^{nR_1}\}\times \{1,\dots,2^{nR_2}\} \rightarrow \hat{\mathcal{Y}}_2^{*n},
\end{align*}
which satisfy 
\begin{align*}
\mathbb{E}~d^*_1(Y_1^n,\hat{Y}_1^n) &\leq D_1+\epsilon\\
\mathbb{E}~d^*_2(Y_2^n,\hat{Y}_2^n) &\leq D_2+\epsilon,
\end{align*}
where 
\begin{align*}
\hat{Y}_1^n &=\phi_1(f_1(Y_1^n),f_2(Y_2^n))\\
\hat{Y}_2^n &=\phi_2(f_1(Y_1^n),f_2(Y_2^n)).
\end{align*}
\end{definition}

%Recall from Definition \ref{defn:BTIB}, that $\mathcal{RD}^i$ consists of all rate-distortion tuples achievable by the Berger-Tung compression scheme, where the reproductions are restricted to the set of product distributions.

%\begin{definition}
%Let $\mathcal{RD}^i$ consist of all rate-distortion tuples achievable by the Berger-Tung compression scheme, where the reproductions are restricted to the set of product distributions.
%\end{definition}

\begin{theorem}\label{thm:strongConverse}
If $(R_1,R_2,D_1,D_2)$ is sequence-achievable, then $(R_1,R_2,D_1,D_2)\in \mathcal{RD}^i=\overline{\mathcal{RD}}^{\star}$.
\end{theorem}
\begin{proof}
The theorem is an immediate consequence of Theorem \ref{thm:MTSCregionDesc} and Lemmas \ref{lem:seqOB} and \ref{lem:OBeqIB}, which are given below.
\end{proof}
\begin{remark}
We refer to Theorem \ref{thm:strongConverse} as the ``strengthened converse" of Theorem \ref{thm:MTSCregionDesc}.  Indeed, it states that enlarging the set of possible reproduction sequences to include non-product distributions cannot attain better performance than when the decoder is restricted to choosing a reproduction sequence from the set of product distributions.  
\end{remark}
\begin{lemma}\label{lem:seqOB}
If $(R_1,R_2,\tilde{D}_1,D_2)$ is sequence-achievable, then there exists a joint distribution 
\begin{align*}
p(y_1,y_2,u_1,u_2,q) =p(q)p(y_1,y_2)p(u_1|y_1,q)p(u_2|y_2,q)
\end{align*}
and a $D_1\leq \tilde{D}_1$ which satisfies
\begin{align*}
D_1 &\geq H(Y_{1}|U_{1},U_{2},Q) \\
D_2 &\geq D_1 +H(Y_{2}|U_{1},U_{2},Q) - H(Y_{1}|U_{1},U_{2},Q),
\end{align*}
and
\begin{align*}
R_1&\geq  H(Y_{1}|U_{2},Q) - D_1\\
R_2&\geq I(Y_{2};U_{2}|Y_{1},Q) + H(Y_{1}|U_{1},Q)-D_1\\
R_1+R_2&\geq I(Y_{2};U_{2}|Y_{1},Q) + H(Y_{1})-D_1.
\end{align*}
\end{lemma}
\begin{proof}
For convenience, let $F_1=f_1(Y_1^n)$ and $F_2=f_2(Y_2^n)$, where $f_1,f_2$ are the encoding functions corresponding to a scheme which achieves $(R_1,R_2,\tilde{D}_1,D_2)$ (in the sequence-reproduction sense). Define $D_1=\frac{1}{n}H(Y_1^n|F_1,F_2)$, so that:
\begin{align}
nD_1=H(Y_1^n|F_1,F_2). \label{eqn:D1ident}
\end{align}
Since $n\tilde{D}_1\geq H(Y_1^n|F_1,F_2)$ by the strengthened version\footnote{See the comment in Section \ref{subsec:ceoStrongerConverse}.} of Lemma \ref{lem:minDistortion}, we have $D_1\leq \tilde{D}_1$ as desired.  By definition of $D_1$, we immediately obtain the following inequality:
\begin{align}
nD_1=\sum_{i=1}^n H(Y_{1,i}|F_1,F_2,Y_{1,i+1}^n) \geq \sum_{i=1}^n H(Y_{1,i}|F_1,F_2,Y_2^{i-1},Y_{1,i+1}^n). \label{eqn:D1LB}
\end{align}
%
%\begin{claim}
%\begin{align*}
%nD_2 \geq nD_1 + \sum_{i=1}^n H(Y_{2,i}|F_1,F_2,Y_2^{i-1},Y_{1,i+1}^n) - H(Y_{1,i}|F_1,F_2,Y_2^{i-1},Y_{1,i+1}^n)
%\end{align*}
%\end{claim}
%\begin{proof}
Next, recall the Csisz\'{a}r sum identity:
\begin{align*}
\sum_{i=1}^n I(Y_{1,i+1}^n;Y_{2,i}|Y_2^{i-1},F_1,F_2) = \sum_{i=1}^n I(Y_{2}^{i-1};Y_{1,i}|Y_{1,i+1}^n,F_1,F_2).
\end{align*}
This, together with \eqref{eqn:D1ident}, implies the following inequality:
\begin{align}
nD_2 \geq nD_1 + \sum_{i=1}^n H(Y_{2,i}|F_1,F_2,Y_2^{i-1},Y_{1,i+1}^n) - H(Y_{1,i}|F_1,F_2,Y_2^{i-1},Y_{1,i+1}^n), \label{eqn:D1D2ineq}
\end{align}
which we can verifiy as follows:
\begin{align*}
nD_2 &\geq H(Y_2^n|F_1,F_2) = \sum_{i=1}^n H(Y_{2,i}|F_1,F_2,Y_2^{i-1})\\
&=\sum_{i=1}^n H(Y_{2,i}|F_1,F_2,Y_2^{i-1},Y_{1,i+1}^n)+I(Y_{1,i+1}^n;Y_{2,i}|F_1,F_2,Y_2^{i-1})\\
&=\sum_{i=1}^n H(Y_{2,i}|F_1,F_2,Y_2^{i-1},Y_{1,i+1}^n)+I(Y_{2}^{i-1};Y_{1,i}|Y_{1,i+1}^n,F_1,F_2)\\
&=H(Y_{1}^n|F_1,F_2) + \sum_{i=1}^n H(Y_{2,i}|F_1,F_2,Y_2^{i-1},Y_{1,i+1}^n) - H(Y_{1,i}|F_1,F_2,Y_2^{i-1},Y_{1,i+1}^n)\\
&=nD_1 + \sum_{i=1}^n H(Y_{2,i}|F_1,F_2,Y_2^{i-1},Y_{1,i+1}^n) - H(Y_{1,i}|F_1,F_2,Y_2^{i-1},Y_{1,i+1}^n).
\end{align*}
Next, observe that we can  lower bound $R_1$ as follows:
\begin{align}
nR_1 &\geq H(F_1) \geq I(Y_1^n;F_1|F_2) \notag \\
&=\sum_{i=1}^n H(Y_{1,i}|F_2,Y_{1}^{i-1}) - H(Y_1^n|F_1,F_2) \notag\\
&\geq \sum_{i=1}^n H(Y_{1,i}|F_2,Y_{1}^{i-1},Y_2^{i-1}) - nD_1 \label{eqn:R1ineq1}\\
&= \sum_{i=1}^n H(Y_{1,i}|F_2,Y_2^{i-1}) - nD_1 \label{eqn:R1ineq2}\\
&\geq \sum_{i=1}^n H(Y_{1,i}|F_2,Y_2^{i-1},Y_{1,i+1}^n) - nD_1. \label{eqn:R1LB}
\end{align}
In the above string of inequalities, \eqref{eqn:R1ineq1} follows from \eqref{eqn:D1ident} and the fact that conditioning reduces entropy.  Equality \eqref{eqn:R1ineq2} follows since  $Y_{1,i}\leftrightarrow F_2,Y_{2}^{i-1}\leftrightarrow Y_1^{i-1}$ form a Markov chain (in that order).

Next, we can obtain a lower bound on $R_2$:
\begin{align}
nR_2 &\geq H(F_2) \geq H(F_2|F_1) = H(F_2|F_1,Y_1^n) + I(Y_1^n;F_2|F_1)\notag\\
&\geq I(Y_2^n;F_2|F_1,Y_1^n) + I(Y_1^n;F_2|F_1)\notag\\
&= I(Y_2^n;F_2|Y_1^n) + I(Y_1^n;F_2|F_1)\label{eqn:R2ineq0}\\
&=\sum_{i=1}^nI(Y_{2,i};F_2|Y_1^n,Y_2^{i-1}) + H(Y_{1,i}|F_1,Y_{1,i+1}^n)-nD_1\label{eqn:R2ineq1}\\
&\geq\sum_{i=1}^nI(Y_{2,i};F_2|Y_1^n,Y_2^{i-1}) + H(Y_{1,i}|F_1,Y_2^{i-1},Y_{1,i+1}^n)-nD_1\notag\\
&=\sum_{i=1}^nI(Y_{2,i};F_2,Y_1^{i-1},Y_2^{i-1}|Y_{1,i},Y_2^{i-1},Y_{1,i+1}^n) + H(Y_{1,i}|F_1,Y_2^{i-1},Y_{1,i+1}^n)-nD_1\label{eqn:R2ineq2}\\
&\geq\sum_{i=1}^nI(Y_{2,i};F_2,Y_2^{i-1}|Y_{1,i},Y_2^{i-1},Y_{1,i+1}^n) + H(Y_{1,i}|F_1,Y_2^{i-1},Y_{1,i+1}^n)-nD_1. \label{eqn:R2LB}
\end{align}
In the above string of inequalities, \eqref{eqn:R2ineq1} follows from \eqref{eqn:D1ident} and the chain rule. \eqref{eqn:R2ineq2} follows from the i.i.d.~property of the sources, and \eqref{eqn:R2LB} follows by monotonicity of mutual information.

A lower bound on the sum-rate $R_1+R_2$ can be obtained as follows:
\begin{align}
n(R_1+R_2) &\geq H(F_1)+H(F_2) \geq H(F_2)+H(F_1|F_2) \notag\\
&\geq I(F_2;Y_1^n,Y_2^n) + I(F_1;Y_1^n|F_2)\notag\\
&=I(F_2;Y_1^n) + I(F_2;Y_2^n|Y_1^n) + I(F_1;Y_1^n|F_2)\notag\\
&=I(F_2;Y_2^n|Y_1^n) + I(F_1,F_2;Y_1^n)\notag\\
&\geq\sum_{i=1}^nI(Y_{2,i};F_2,Y_2^{i-1}|Y_{1,i},Y_2^{i-1},Y_{1,i+1}^n) + H(Y_{1,i})-nD_1. \label{eqn:R1R2LB}
\end{align}
Where \eqref{eqn:R1R2LB} follows in a manner similar to \eqref{eqn:R2ineq0}-\eqref{eqn:R2LB} in the lower bound on $R_2$.

%Also, since $Y_1^{i-1}\leftrightarrow (Y_2^{i-1},Y_{1,i+1}^n,F_2) \leftrightarrow Y_{1,i}$ form a Markov chain, we have:
%\begin{align*}
%\sum_{i=1}^n H(Y_{1,i}|F_1,F_2,Y_2^{i-1}) \
%\end{align*}

Now, define $U_{1,i}\triangleq F_1$, $U_{2,i}\triangleq (F_2,Y_2^{i-1})$, and $Q_{i}\triangleq (Y_2^{i-1},Y_{1,i+1}^n)$.  Then we can summarize our results so far as follows.  Inequalities \eqref{eqn:D1LB} and \eqref{eqn:D1D2ineq} become
\begin{align*}
D_1 &\geq \frac{1}{n}\sum_{i=1}^nH(Y_{1,i}|U_{1,i},U_{2,i},Q_i)\\
D_2 &\geq D_1 +\frac{1}{n}\sum_{i=1}^n H(Y_{2,i}|U_{1,i},U_{2,i},Q_i) - H(Y_{1,i}|U_{1,i},U_{2,i},Q_i),
\end{align*}
and inequalities \eqref{eqn:R1LB}, \eqref{eqn:R2LB}, and   \eqref{eqn:R1R2LB} can be written as:
\begin{align*}
R_1&\geq \frac{1}{n}\sum_{i=1}^n H(Y_{1,i}|U_{2,i},Q_i) - D_1\\
R_2&\geq \frac{1}{n}\sum_{i=1}^nI(Y_{2,i};U_{2,i}|Y_{1,i},Q_i) + H(Y_{1,i}|U_{1,i},Q_i)-D_1\\
R_1+R_2&\geq \frac{1}{n}\sum_{i=1}^nI(Y_{2,i};U_{2,i}|Y_{1,i},Q_i) + H(Y_{1,i})-D_1.
\end{align*}

Next, we note that $U_{1,i}\leftrightarrow Y_{1,i}\leftrightarrow Y_{2,i}\leftrightarrow U_{2,i}$ form a Markov chain (in that order) conditioned on $Q_i$.  Moreover, $Q_i$ is independent of $Y_{1,i},Y_{2,i}$. Hence, a standard timesharing argument proves the lemma.
\end{proof}

\begin{lemma}\label{lem:OBeqIB}
Fix $(R_1,R_2,D_1,D_2)$. If there exists a joint distribution of the form
\begin{align*}
p(y_1,y_2,u_1,u_2,q) =p(q)p(y_1,y_2)p(u_1|y_1,q)p(u_2|y_2,q)
\end{align*}
which satisfies
\begin{align}
D_1 &\geq H(Y_{1}|U_{1},U_{2},Q) \label{eqn:D1LBhyp}\\
D_2 &\geq D_1 +H(Y_{2}|U_{1},U_{2},Q) - H(Y_{1}|U_{1},U_{2},Q),\label{eqn:D1D2hyp}
\end{align}
and
\begin{align}
R_1&\geq  H(Y_{1}|U_{2},Q) - D_1 \label{eqn:R1}\\
R_2&\geq I(Y_{2};U_{2}|Y_{1},Q) + H(Y_{1}|U_{1},Q)-D_1\label{eqn:R2}\\
R_1+R_2&\geq I(Y_{2};U_{2}|Y_{1},Q) + H(Y_{1})-D_1,\label{eqn:R1R2}
\end{align}
then $(R_1,R_2,D_1,D_2)\in\mathcal{RD}^i$. 
\end{lemma}
\begin{proof}
Let $\mathcal{P}$ denote the polytope of rate pairs which satisfy the inequalities \eqref{eqn:R1}-\eqref{eqn:R1R2}. It suffices to show that if $(r_1,r_2)$ is a vertex of $\mathcal{P}$, then $(r_1,r_2,D_1,D_2)\in\mathcal{RD}^i$. For convenience, let $[x]^+=\max\{x,0\}$. There are only two extreme points of $\mathcal{P}$:
\begin{align*}
r_1^{(1)} &= \bigg[H(Y_{1}|U_{2},Q) - D_1\bigg]^+\\
r_2^{(1)} &= I(Y_{2};U_{2}|Y_{1},Q) + H(Y_{1})-D_1-r_1^{(1)},
\end{align*}
and 
\begin{align*}
r_1^{(2)} &=  I(Y_{2};U_{2}|Y_{1},Q) + H(Y_{1})-D_1-r_2^{(2)},\\
r_2^{(2)} &=\bigg[I(Y_{2};U_{2}|Y_{1},Q) + H(Y_{1}|U_{1},Q)-D_1\bigg]^+.
\end{align*}

We first analyze the extreme point $(r_1^{(1)},r_2^{(1)})$:
\begin{itemize}
\item  \textbf{ Case 1.1: $r_1^{(1)}=0$}.  In this case, we have $r_2^{(1)} = I(Y_{2};U_{2}|Y_{1},Q) + H(Y_{1})-D_1$.  This can be expressed as:
\begin{align*}
r_2^{(1)}=(1-\theta) I(Y_2;U_2|Q),
\end{align*}
where 
\begin{align*}
\theta &= \frac{D_1 -I(Y_{2};U_{2}|Y_{1},Q) - H(Y_{1})+I(Y_2;U_2|Q)}{I(Y_2;U_2|Q)}.
\end{align*}
Since $r_1^{(1)}=0$, we must have $D_1\geq H(Y_1|U_2,Q)$.  This implies that 
\begin{align*}
\theta \geq \frac{H(Y_1|U_2,Q) -I(Y_{2};U_{2}|Y_{1},Q) - H(Y_{1})+I(Y_2;U_2|Q)}{I(Y_2;U_2|Q)}=0.
\end{align*}
Also, we can assume without loss of generality that $D_1\leq H(Y_1)$, hence $\theta\in [0,1]$.  Applying the Berger-Tung achievability scheme, we can achieve the following distortions:
\begin{align}
D_1^{\theta} &= \theta H(Y_1) + (1-\theta)H(Y_1|U_2,Q)\notag\\
&=H(Y_1|U_2,Q) + \theta I(Y_1;U_2|Q)\notag\\
&\leq H(Y_1|U_2,Q)+ D_1 -I(Y_{2};U_{2}|Y_{1},Q) - H(Y_{1})+I(Y_2;U_2|Q)\label{eqn:DPIc11}\\
&=D_1 -I(Y_{2};U_{2}|Y_{1},Q) - I(Y_{1};U_2|Q)+I(Y_2;U_2|Q)\notag\\
&=D_1,\notag
\end{align}
where \eqref{eqn:DPIc11} follows since $I(Y_1;U_2|Q)\leq I(Y_2;U_2|Q)$ by the data processing inequality.

\begin{align}
D_2^{\theta} &= \theta H(Y_2) + (1-\theta)H(Y_2|U_2,Q)\notag\\
&=H(Y_2|U_2,Q) + \theta I(Y_2;U_2|Q)\notag\\
&= H(Y_2|U_2,Q)+ D_1 -I(Y_{2};U_{2}|Y_{1},Q) - H(Y_{1})+I(Y_2;U_2|Q)\notag\\
&= H(Y_2)+ D_1 -I(Y_{2};U_{2}|Y_{1},Q) - H(Y_{1})\notag\\
&=H(Y_2|Y_1,U_2,Q)+ D_1 -H(Y_{1}|Y_2)\notag\\
&=H(Y_2|Y_1,U_1,U_2,Q)+ D_1 -H(Y_{1}|Y_2)\label{eqn:D2markov}\\
&\leq H(Y_2|Y_1,U_1,U_2,Q)+ D_1 -H(Y_{1}|Y_2,U_1,U_2,Q)\notag\\
&= H(Y_2|U_1,U_2,Q)+ D_1 -H(Y_{1}|U_1,U_2,Q)\notag\\
&\leq D_2,\label{eqn:D2final}
\end{align}
where \eqref{eqn:D2markov} follows since $U_1 \leftrightarrow (Y_1,U_2,Q)\leftrightarrow Y_2$, and  \eqref{eqn:D2final} follows from \eqref{eqn:D1D2hyp}.

\item  \textbf{ Case 1.2: $r_1^{(1)}\geq0$}.  In this case, we have $r_2^{(1)} = I(Y_{2};U_{2}|Y_{1},Q) + I(Y_{1};U_2|Q)=I(Y_2;U_2|Q)$.  Also, we can write  $r_1^{(1)}$ as:
\begin{align*}
r_1^{(1)}=(1-\theta) I(Y_1;U_1|U_2,Q),
\end{align*}
where 
\begin{align*}
\theta &= \frac{D_1 -H(Y_{1}|U_2,Q)+I(Y_1;U_1|U_2,Q)}{I(Y_1;U_1|U_2,Q)}.
\end{align*}
Since $r_1^{(1)}\geq0$, we must have $D_1\leq H(Y_1|U_2,Q)$.  This implies that 
\begin{align*}
\theta \leq \frac{H(Y_1|U_2,Q) -H(Y_{1}|U_2,Q)+I(Y_1;U_1|U_2,Q)}{I(Y_1;U_1|U_2,Q)}=1.
\end{align*}
Also, \eqref{eqn:D1LBhyp} implies that $D_1\geq H(Y_1|U_1,U_2,Q)$, hence $\theta\in [0,1]$.  Applying the Berger-Tung achievability scheme, we can achieve the following distortions:
\begin{align}
D_1^{\theta} &= \theta H(Y_1|U_2,Q) + (1-\theta)H(Y_1|U_1,U_2,Q)\notag\\
&=H(Y_1|U_1,U_2,Q) + \theta I(Y_1;U_1|U_2,Q)\notag\\
&= H(Y_1|U_1,U_2,Q)+ D_1 -H(Y_{1}|U_2,Q)+I(Y_1;U_1|U_2,Q)\notag\\
&=D_1,\notag
\end{align}
and
\begin{align}
D_2^{\theta} &= \theta H(Y_2|U_2,Q) + (1-\theta)H(Y_2|U_1,U_2,Q)\notag\\
&=H(Y_2|U_1,U_2,Q)+\theta I(Y_2;U_1|U_2,Q)\notag\\
&\leq  H(Y_2|U_1,U_2,Q) + D_1 -H(Y_{1}|U_2,Q)+I(Y_1;U_1|U_2,Q) \label{eqn:DPIc12}\\
&=H(Y_2|U_1,U_2,Q)+ D_1 -H(Y_{1}|U_1,U_2,Q)\notag\\
&\leq D_2,\label{eqn:D2final2}
\end{align}
where \eqref{eqn:DPIc12} follows since $I(Y_2;U_1|U_2,Q)\leq I(Y_1;U_1|U_2,Q)$ by the data processing inequality, and  \eqref{eqn:D2final2} follows from \eqref{eqn:D1D2hyp}.
\end{itemize}

In a similar manner, we now analyze the second extreme point $(r_1^{(2)},r_2^{(2)})$:
\begin{itemize}
\item  \textbf{ Case 2.1: $r_2^{(2)}=0$}.  In this case, we have $r_1^{(2)} = I(Y_{2};U_{2}|Y_{1},Q) + H(Y_{1})-D_1$.  This can be expressed as:
\begin{align*}
r_1^{(2)}=(1-\theta) I(Y_1;U_1|Q),
\end{align*}
where 
\begin{align*}
\theta &= \frac{D_1 -I(Y_{2};U_{2}|Y_{1},Q) - H(Y_{1})+I(Y_1;U_1|Q)}{I(Y_1;U_1|Q)}.
\end{align*}
Since $r_2^{(2)}=0$, we must have $D_1\geq H(Y_1|U_1,Q)+I(Y_2;U_2|Y_1,Q)$.  This implies that 
\begin{align*}
\theta \geq \frac{H(Y_1|U_1,Q)+I(Y_2;U_2|Y_1,Q) -I(Y_{2};U_{2}|Y_{1},Q) - H(Y_{1})+I(Y_1;U_1|Q)}{I(Y_1;U_1|Q)}=0.
\end{align*}
Also, we can assume without loss of generality that $D_1\leq H(Y_1)$, hence 
\begin{align*}
\theta &\leq \frac{H(Y_1) -I(Y_{2};U_{2}|Y_{1},Q) - H(Y_{1})+I(Y_1;U_1|Q)}{I(Y_1;U_1|Q)}\leq 1,
\end{align*}
and therefore $\theta\in [0,1]$.  Applying the Berger-Tung achievability scheme, we can achieve the following distortions:
\begin{align}
D_1^{\theta} &= \theta H(Y_1) + (1-\theta)H(Y_1|U_1,Q)\notag\\
&=H(Y_1|U_1,Q) + \theta I(Y_1;U_1|Q)\notag\\
&= H(Y_1|U_1,Q)+ D_1 -I(Y_{2};U_{2}|Y_{1},Q) - H(Y_{1})+I(Y_1;U_1|Q) \notag\\
&=D_1 -I(Y_{2};U_{2}|Y_{1},Q)\notag\\
&\leq D_1,\notag
\end{align}
%where \eqref{eqn:DPIc11} follows since $I(Y_1;U_2|Q)\leq I(Y_2;U_2|Q)$ by the data processing inequality.
and
\begin{align}
D_2^{\theta} &= \theta H(Y_2) + (1-\theta)H(Y_2|U_1,Q)\notag\\
&=H(Y_2|U_1,Q) + \theta I(Y_2;U_1|Q)\notag\\
&\leq H(Y_2|U_1,Q)+ D_1 -I(Y_{2};U_{2}|Y_{1},Q) - H(Y_{1})+I(Y_1;U_1|Q) \label{eqn:DPIc21}\\
&=H(Y_2|Y_1,U_2,Q)+ D_1 -H(Y_{1}|Y_2,U_1,Q)\notag\\
&=H(Y_2|Y_1,U_1,U_2,Q)+ D_1 -H(Y_{1}|Y_2,U_1,U_2,Q)\label{eqn:D2markov2}\\
&= H(Y_2|U_1,U_2,Q)+ D_1 -H(Y_{1}|U_1,U_2,Q)\notag\\
&\leq D_2,\label{eqn:D2final21}
\end{align}
where \eqref{eqn:DPIc21} follows since $I(Y_2;U_1|Q)\leq I(Y_1;U_1|Q)$ by the data processing inequality,
 \eqref{eqn:D2markov2} follows since $U_1 \leftrightarrow (Y_1,U_2,Q)\leftrightarrow Y_2$ and $U_2 \leftrightarrow (Y_2,U_1,Q)\leftrightarrow Y_1$, and  \eqref{eqn:D2final21} follows from \eqref{eqn:D1D2hyp}.

\item  \textbf{ Case 2.2: $r_2^{(2)}\geq0$}.  In this case, we have $r_1^{(2)} = I(Y_1;U_1|Q)$.  Also, we can write  $r_2^{(2)}$ as:
\begin{align*}
r_2^{(2)}=(1-\theta) I(Y_2;U_2|U_1,Q),
\end{align*}
where 
\begin{align*}
\theta &= \frac{D_1 -H(Y_{1}|U_1,Q)-I(Y_2;U_2|Y_1,Q)+ I(Y_2;U_2|U_1,Q)}{I(Y_2;U_2|U_1,Q)}.
\end{align*}
Since $r_2^{(2)}\geq0$, we must have $D_1\leq H(Y_1|U_1,Q)+I(Y_2;U_2|Y_1,Q)$.  This implies that $\theta\leq 1$.  Also, \eqref{eqn:D1LBhyp} implies that $D_1\geq H(Y_1|U_1,U_2,Q)$, yielding
\begin{align*}
\theta \geq \frac{H(Y_1|U_1,U_2,Q) -H(Y_{1}|U_1,Q)-I(Y_2;U_2|Y_1,Q)+ I(Y_2;U_2|U_1,Q)}{I(Y_2;U_2|U_1,Q)}=0.
\end{align*}
Therefore, $\theta\in [0,1]$.  Applying the Berger-Tung achievability scheme, we can achieve the following distortions:
\begin{align}
D_1^{\theta} &= \theta H(Y_1|U_1,Q) + (1-\theta)H(Y_1|U_1,U_2,Q)\notag\\
&=H(Y_1|U_1,U_2,Q) + \theta I(Y_1;U_2|U_1,Q)\notag\\
&\leq H(Y_1|U_1,U_2,Q)+ D_1 -H(Y_{1}|U_1,Q)-I(Y_2;U_2|Y_1,Q)+ I(Y_2;U_2|U_1,Q) \label{eqn:DPIc22}\\
&=D_1,\notag
\end{align}
where \eqref{eqn:DPIc22} follows since $I(Y_1;U_2|U_1,Q)\leq I(Y_2;U_2|U_1,Q)$ by the data processing inequality.
\begin{align}
D_2^{\theta} &= \theta H(Y_2|U_1,Q) + (1-\theta)H(Y_2|U_1,U_2,Q)\notag\\
&=H(Y_2|U_1,U_2,Q)+\theta I(Y_2;U_2|U_1,Q)\notag\\
&=  H(Y_2|U_1,U_2,Q) + D_1 -H(Y_{1}|U_1,Q)-I(Y_2;U_2|Y_1,Q)+ I(Y_2;U_2|U_1,Q)  \notag\\
&=H(Y_2|U_1,U_2,Q)+ D_1 -H(Y_{1}|U_1,U_2,Q)\notag\\
&\leq D_2,\label{eqn:D2final22}
\end{align}
where  \eqref{eqn:D2final22} follows from \eqref{eqn:D1D2hyp}.
\end{itemize}

Thus, this proves that the Berger-Tung compression scheme can achieve any rate distortion tuple $(r_1,r_2,D_1,D_2)$ for $(r_1,r_2)\in\mathcal{P}$.  Since $\mathcal{RD}^i$ is, by definition, the set of rate distortion tuples attainable by the Berger-Tung achievability scheme, we must have that $(R_1,R_2,D_1,D_2)\in \mathcal{RD}^i$.  This proves the lemma.
\end{proof}

\section{A Lemma for the Daily Double}\label{app:dailyDoubleLemma}

For a given joint distribution $p(y_1,y_2)$ on the finite alphabet $\mathcal{Y}_1\times \mathcal{Y}_2$, let $\mathcal{P}(R_1,R_2)$ denote the set of joint pmf's of the form
\begin{align*}
p(q,y_1,y_2,u_1,u_2)=p(q)p(y_1,y_2)p(u_1|y_1,q)p(u_1|y_1,q)
\end{align*}
which satisfy
\begin{align*}
R_1 &\geq I(Y_1;U_1|U_2,Q) \\
R_2 &\geq I(Y_2;U_2|U_1,Q)\\
R_1+R_2 &\geq I(Y_1,Y_2;U_1,U_2|Q)
\end{align*}
for given finite alphabets $\mathcal{U}_1,\mathcal{U}_2,\mathcal{Q}$.

\begin{lemma} \label{lem:dailyDouble}
For $R_1,R_2$ satisfying $R_1\leq H(Y_1)$, $R_2\leq H(Y_2)$, and $R_1+R_2 \leq H(Y_1,Y_2)$, the infimum
\begin{align*}
\inf_{p\in\mathcal{P}(R_1,R_2)} \left\{ H(Y_1|U_1,U_2,Q)+H(Y_2|U_1,U_2,Q)\right\}
\end{align*}
is attained by some $p^*\in \mathcal{P}(R_1,R_2)$ which satisfies $R_1+R_2 = I(Y_1,Y_2;U_1^*,U_2^*|Q^*)$, where $U_1^*,U_2^*,Q^*$ correspond to the auxiliary random variables defined by $p^*$.
\end{lemma}
\begin{proof}
First, note that the infimum is always attained since $\mathcal{P}(R_1,R_2)$ is compact and the objective function is continuous on $\mathcal{P}(R_1,R_2)$.  Therefore, let $U_1^*,U_2^*,Q^*$ correspond to the auxiliary random variables which attain the infimum.

If $H(Y_1|U_1^*,U_2^*,Q^*)+H(Y_2|U_1^*,U_2^*,Q^*)=0$, then we must have \linebreak $I(Y_1,Y_2;U_1^*,U_2^*|Q^*)=H(Y_1,Y_2)$.  Thus, $R_1+R_2 = I(Y_1,Y_2;U_1^*,U_2^*|Q^*)$.

Next, consider the case where $H(Y_1|U_1^*,U_2^*,Q^*)+H(Y_2|U_1^*,U_2^*,Q^*)>0$.  Assume for sake of contradiction that $R_1+R_2 > I(Y_1,Y_2;U_1^*,U_2^*|Q^*)$.  For any $p\in\mathcal{P}(R_1,R_2)$:
\begin{align*}
I(Y_1;U_1|U_2,Q) + I(Y_2;U_2|U_1,Q) \leq  I(Y_1,Y_2;U_1,U_2|Q).
\end{align*}
Hence, at most one of the remaining rate constraints can be satisfied with equality.  If none of the rate constraints are satisfied with equality, then define
\begin{align*}
(\tilde{U}_1,\tilde{U}_2)=\left\{ 
\begin{array}{ll}
(U_1^*,U_2^*) & \mbox{with probability $1-\epsilon$}\\
(Y_1,Y_2) & \mbox{with probability $\epsilon$.}
\end{array}
\right.
\end{align*}
For $\epsilon>0$ sufficiently small, the distribution $\tilde{p}$ corresponding to the auxiliary random variables $\tilde{U}_1,\tilde{U}_2,Q^*$ is still in $\mathcal{P}(R_1,R_2)$.  However, $\tilde{p}$ satisfies
\begin{align*}
H(Y_1|\tilde{U}_1,\tilde{U}_2,Q^*)+H(Y_2|\tilde{U}_1,\tilde{U}_2,Q^*)<H(Y_1|U_1^*,U_2^*,Q^*)+H(Y_2|U_1^*,U_2^*,Q^*),
\end{align*}
which contradicts the optimality of $p^*$.

Therefore, assume without loss of generality that 
\begin{align*}
R_1 &= I(Y_1;U_1^*|U_2^*,Q^*) \\
R_1+R_2 &> I(Y_1,Y_2;U_1^*,U_2^*|Q^*).
\end{align*}
This implies that $R_2>I(Y_2;U_2^*|Q^*)$.  Now, define 
\begin{align*}
\tilde{U}_2=\left\{ 
\begin{array}{ll}
U_2^* & \mbox{with probability $1-\epsilon$}\\
Y_2 & \mbox{with probability $\epsilon$.}
\end{array}
\right.
\end{align*}
Note that for $\epsilon>0$ sufficiently small:
\begin{align*}
I(Y_2;U_2^*|Q^*) &< I(Y_2;\tilde{U}_2|Q^*) < R_2\\
I(Y_1,Y_2;U_1^*,U_2^*|Q^*) &< I(Y_1,Y_2;U_1^*,\tilde{U}_2|Q^*)<R_1+R_2,
\end{align*}
and for any $\epsilon\in[0,1]$:
\begin{align}
R_1=I(Y_1;U_1^*|U_2^*,Q^*) &\geq I(Y_1;U_1^*|\tilde{U}_2,Q^*)\notag \\
H(Y_1|U_1^*,U_2^*,Q^*)+H(Y_2|U_1^*,U_2^*,Q^*) &\geq H(Y_1|U_1^*,\tilde{U}_2,Q^*)+H(Y_2|U_1^*,\tilde{U}_2,Q^*). \label{eqn:objNonInc}
\end{align}
Since $R_2\leq H(Y_2)$, as $\epsilon$ is increased from $0$ to $1$, at least one of the following must occur:
\begin{enumerate}
\item $I(Y_2;\tilde{U}_2|Q^*) = R_2$.
\item $I(Y_1,Y_2;U_1^*,\tilde{U}_2|Q^*)=R_1+R_2$.
\item $I(Y_1;U_1|\tilde{U}_2,Q^*) < R_1$.
\end{enumerate}
If either of events 1 or 2 occur first then the sum-rate constraint is met with equality (since they are equivalent in this case).  If event 3 occurs first, then all rate constraints are satisfied with strict inequality and we can apply the above argument to contradict optimality of $p^*$.  Since \eqref{eqn:objNonInc} shows that the objective is nonincreasing in $\epsilon$, there must exist a $\tilde{p}\in\mathcal{P}(R_1,R_2)$ which attains the infimum and satisfies the sum-rate constraint with equality.  
\end{proof}

\bibliographystyle{ieeetr}
\bibliography{mybib}

\end{document}